\documentclass[a4paper,onecolumn,11pt]{quantumarticle}
\pdfoutput=1
\usepackage[utf8]{inputenc}
\usepackage[english]{babel}
\usepackage[T1]{fontenc}
\usepackage[numbers]{natbib}
\usepackage{listings}
\usepackage{hyperref}
\usepackage[resetlabels]{multibib}
\usepackage{multirow}
\usepackage{relsize}

\usepackage{amsmath}
\usepackage{tikz}
\usepackage{lipsum}
\usepackage{graphicx}
\usepackage{dcolumn}
\usepackage{bm}%

\usepackage{amssymb}
\usepackage{amsthm}
\usepackage[caption=false,font=normalsize,labelfont=sf,textfont=sf]{subfig}
\usepackage{xcolor}
\usepackage{verbatim}
\usepackage{soul}
\usepackage{algorithm,algpseudocode}
\usepackage{mathtools}
\usepackage{array}
\usepackage{textcomp}
\usepackage{multirow}

\usepackage{balance}
\usepackage{array}
\usepackage{url}
\usepackage{longtable, colortbl}

\newtheorem{theorem}{Theorem}
\newtheorem{lemma}{Lemma}

\newtheorem{definition}{Definition}

\newtheorem{example}{Example}
\newcommand{\bc}[1]{{\textcolor{black}{#1}}}

\def\bG{\mathbf{G}}

\begin{document}
    \title{Transversal Code Switching in Triorthogonal Code Families for Fault-Tolerant Universal Quantum Computing}


\author[1]{Dawei Jiao}
\email{dawei.jiao@aalto.fi}
\author[1]{Mahdi Bayanifar}
\author[2]{Alexei Ashikhmin}
\author[1]{Olav Tirkkonen}

\affiliation[1]{Department of Communications and Networking, Aalto University, Finland}
\affiliation[2]{Nokia Bell Labs, Murray Hill, New Jersey, USA}

\begin{abstract}

We study transversal code switching in triorthogonal codes to enable universal fault-tolerant quantum computation. First, we establish a CNOT transversality condition for switching between a triorthogonal code and a symmetric Calderbank–Shor–Steane (CSS) code with different numbers of physical qubits. For punctured quantum Reed–Muller (PQRM) codes, we introduce a code map that identifies compatible symmetric CSS paired codes admitting one-way transversal CNOT gates. This extends code-switching techniques beyond conventional color-code architectures to codes supporting higher levels of the Clifford hierarchy. For triorthogonal codes outside the PQRM family, we provide a systematic method for constructing compatible symmetric CSS codes. We further show that our switching protocol integrates directly into Steane error correction, reducing ancilla overhead. Finally, we benchmark the framework by compiling logical \(n\)-qubit quantum Fourier transform circuits using the \([[15,1,3]]\) and \([[31,1,3]]\) codes. Our results demonstrate that exploiting higher level Clifford hierarchy gates through code switching reduces the total CNOT count for logical QFT circuits compared with implementations using gates from the third level of the Clifford hierarchy. These findings highlight the potential of heterogeneous code architectures for reducing the resource costs of fault-tolerant quantum computation.
\end{abstract}


\maketitle

\section{Introduction}
Quantum error correcting codes (QECCs) are crucial for suppressing physical errors in large quantum systems, through encoding information into a protected subspace of a larger Hilbert space~\cite{nielsen2002Book,chuang1995quantumerrorcorrectioncoding,laflamme1996perfectquantumerrorcorrection}.
To ensure reliable computation, quantum logical operations must be implemented fault-tolerantly, so that a single physical error does not spread uncontrollably and become unrecoverable~\cite{Gottesman_1998,preskill1997faulttolerantquantumcomputation}. However, fault-tolerance alone is not sufficient, universal quantum computation also requires a gate set capable of implementing arbitrary unitary operations. Transversal gates are highly desirable due to their inherent error-localization properties, but as shown by the Eastin-Knill theorem~\cite{eastin2009restrictions}, no QECC supports a universal gate set composed entirely of transversal operations.

To circumvent this limitation, several strategies have been proposed. The standard approach uses magic states~\cite{bravyi2005magicstate,li2015magic}, where non-Clifford gates, such as the $\mathbf{\rm T}$-gate, are implemented by gate teleportation using specially prepared ancillary states. When combined with codes supporting transversal Clifford gates this enables universal fault-tolerant computation. However, preparing high-fidelity magic states through distillation is resource costly, often requiring orders of magnitude more qubits and gates than other logical operations~\cite{fowler2013surface}.

An alternative strategy focuses on codes that have transversal non-Clifford gates, such as triorthogonal codes with transversal $\mathbf{\rm T}$ and controlled-controlled-Z (CCZ) gates~\cite{bravyi2005magicstate,narayana2020optimality}, 
while using additional techniques to implement non-transversal Clifford gates fault-tolerantly~\cite{paetznick2013universal,anderson2014codeswitch,yoder2017universal,banfield2022rewiring}. 
In this paper, we adopt the latter approach and propose a transversal code switching framework for triorthogonal codes. Depending on the code structure, we develop different strategies to construct switchable code pairs. Specifically, for triorthogonal codes that can be classified as punctured quantum Reed–Muller (PQRM) codes, we introduce a \emph{code map} to systematically identify compatible symmetric codes that admit a one-way transversal CNOT gate with a given code. This protocol generalizes previous one-way transversal CNOT techniques for color codes~\cite{Heu_en_2025TransversalCodeSwithcing} to the PQRM family, thereby enabling transversal code switching for codes supporting transversal gates at higher levels of the Clifford hierarchy.
For the rest of triorthogonal codes, which cannot be parameterized as PQRM codes. We provide a systematic way to construct a symmetric Calderbank-Shor-Steane (CSS) code from a given triorthogonal code, thus completing the transversal code switching strategy for all kinds of triorthogonal codes. 
Furthermore, we combine the transversal code switching protocol with the Steane error correction protocol, reducing resource overhead.
Using the $[[15,1,3]]$ and $[[31,1,3]]$ triorthogonal codes in $n$-qubit quantum Fourier Transform (QFT) circuits as an example, we compare the overall CNOT gate count for using these different codes. The results demonstrate that, even considering the code switching protocol, using the high-level Clifford hierarchy gate set transversal codes could reduce the gate cost for QFT logical circuits compared to lower-level codes.

The rest of the paper is organized as follows: Section~\ref{Sec.Pre} introduces preliminaries on CSS codes, transversal gates, triorthogonal codes, and the PQRM code family.  In In Section~\ref{Sec.CodeGeneration}, we prove the transversality conditions for PQRM codes and present systematic constructions of paired symmetric CSS codes for general triorthogonal codes, illustrating these procedures using the $[[31,1,3]]$ and $[[14,2,2]]$ codes.
Section~\ref{Sec.CircuitStructure} details the state teleportation circuits used to transfer logical states between code pairs and demonstrates their integration with Steane-type error correction. Finally, Section~\ref{Sec.ResourceEST} evaluates the CNOT gate count for compiling $n$-qubit logical QFT circuits via our code switching protocol on 15-qubit and 31-qubit PQRM codes, verifying that switching to codes with transversal gates at higher levels of the Clifford hierarchy substantially reduces overall gate overhead.

\section{Preliminaries}
\label{Sec.Pre}
\subsection{CSS Codes}

CSS quantum error correction codes 
can be constructed based on two classical binary linear codes 
$\mathcal{C}_1\left[n, k_1, d_1\right]$ and $\mathcal{C}_2\left[n, k_2, d_2\right]$, such that $\mathcal{C}_2^{\perp} \subset \mathcal{C}_1$, with 
$\mathcal{C}_2^{\perp}$ the dual space of $\mathcal{C}_2$. For  $\mathbf{G}$ a generator matrix, we denote by $\mathbf{G}^\perp$ a generator matrix of the dual space of $\mathbf{G}$.

A quantum $[[n, k, d]]$ CSS code $\mathcal{Q} = \rm{CSS}\left(\mathcal{C}_1, \mathcal{C}_2\right)$ is defined as a
$2^k$-dimensional linear subspace of $\mathbb{C}^{2^n}$ with orthonormal basis~\cite{nielsen2002Book}
\begin{equation}
	\lvert \bm{\psi} \rangle_L = \frac{1}{\sqrt{\left\lvert \mathcal{C}_2^{\perp}\right\rvert} }\sum_{\mathbf{y} \in \mathcal{C}_2^{\perp}}{| \mathbf{x}_\psi+\mathbf{y} \rangle}, \label{Eq:StablStCSS}
\end{equation}
where $\bm{\psi} \in \mathbb{F}_2^k$ and 
\begin{equation}
    \mathbf{x}_{\bm{\psi}} = \bm{\psi} \mathbf{A}\,
    \label{eq:mappingMat}
\end{equation}
is a codeword vector encoding the logical state $\bm{\psi}$ encoded through the mapping matrix $\mathbf{A}\in \mathbb{F}_2^{k\times n}$. The mapping matrix  is a full-rank matrix which is a generator of the quotient group $\mathcal{C}_1 / \mathcal{C}_2^{\perp}$, i.e., $\mathbf{A}\cong\mathcal{C}_1 / \mathcal{C}_2^{\perp}$.
Note that $k=k_1+k_2-n$ and the  minimum distance of the code is 
\begin{equation}
 d = \min\{\omega_H(\mathbf{c})~\lvert~
        \mathbf{c}\in\mathcal{C}_1/\mathcal{C}_2^\perp ~\cup~\mathcal{C}_2/\mathcal{C}_1^\perp
       \}\,.
       \label{eq:mindist}
\end{equation}

Pauli matrices for a single qubit system are defined as: 
\begin{equation}
    \mathbf{X} \triangleq
    \begin{bmatrix}
        0 & 1 \\
        1 & 0
    \end{bmatrix},
    \;
    \mathbf{Z} \triangleq
    \begin{bmatrix}
        1 & 0 \\
        0 & -1
    \end{bmatrix}
    \;
    \mathbf{Y} \triangleq
    \begin{bmatrix}
        0 & -i \\
        i & 0
    \end{bmatrix},
\end{equation}
where $i\triangleq \sqrt{-1}$. Pauli 
errors acting on $n$ qubits have the form
$\mathbf{E}\left( \mathbf{a, b} \right) = i^{\mathbf{a b}^{\rm T}} \mathbf{D}\left( \mathbf{a, b} \right)$, where $\mathbf{D}\left( \mathbf{a, b} \right) = \mathbf{X}^{a_1}\mathbf{Z}^{b_1} \otimes ... \otimes \mathbf{X}^{a_n}\mathbf{Z}^{b_n}$ and  $\mathbf{a} = \left[a_1, ..., a_n \right]$, $\mathbf{b} = \left[b_1, ..., b_n \right]$ are binary vectors. 
We denote by
$\gamma$ the homomorphism defined by 
$\gamma\left( i^{\kappa} \mathbf{D} \left( \mathbf{a, b} \right) \right) = \left[\mathbf{a}, \mathbf{b}\right]$, for $\kappa\in \left\{0,1,2,3 \right\}$. 

A ${\rm CSS}(\mathcal{C}_1,\mathcal{C}_2)$  
code, defined above, is an $[[n,k]]$ stabilizer code~\cite{nielsen2002Book} whose $n-k$ generators correspond to binary vectors $\mathbf{g}_i$, $i=1,...,n-k$ in the form $[\mathbf{a},\mathbf{0}]$ or $[\mathbf{0},\mathbf{b}]$, where 
$\mathbf{0}\in \mathbb{F}_2^n$ is the all zero vector. Combining vectors $\mathbf{a}$ and $\mathbf{b}$, we obtain parity check matrices $\mathbf{H}\left(\mathcal{C}_2 \right)$ and $\mathbf{H}\left(\mathcal{C}_1 \right)$ of codes $\mathcal{C}_2$ and $\mathcal{C}_1$ respectively. Putting these matrices together, we obtain the matrix $\mathbf{G}^{\mathcal{Q}}$, the matrix of stabilizer generators of $\rm CSS(\mathcal{C}_1,\mathcal{C}_2)$:

\begin{equation}
\mathbf{G}^{\mathcal{Q}} = 
\left[
\begin{array}{c|c}
\mathbf{H}\left(\mathcal{C}_2 \right) & \mathbf{0} \\  
\hline
\mathbf{0} & \mathbf{H}\left(\mathcal{C}_1 \right)
\end{array}
\right].
\end{equation}
The binary matrix $\mathbf{G}^{\mathcal{Q}}$ has dimension $\left( n-k \right)\times2n$, which implies that $\mathcal{Q}$ encodes $k$ logical qubits into $n$ physical qubits. 
Since the $\mathcal{C}_1$ and $\mathcal{C}_2$ only contribute to $Z$ and $X$ parts separately, we can also name $\mathcal{C}_1$ as $\mathcal{C}_{\rm Z}$ and $\mathcal{C}_2$ as $\mathcal{C}_{\rm X}$.
A CSS code is called {\em symmetric} if $\mathcal{C}_1 = \mathcal{C}_2$.

\subsection{Transversality}

Let $\mathbf{U}$ be a $2^m\times 2^m$ operator acting on $m$ qubits. The concept from the literature~\cite{Gottesman_1998,preskill1997faulttolerantquantumcomputation,bravyi2005magicstate,narayana2020optimality} of a code allowing a transversal implementation of $\mathbf{U}$ can be formalized as a definition of a {\em U-transversal code}  as follows.
Let $\mathbf{\hat {U}}_i$ be the $2^{mn}\times 2^{mn}$ operator that applies $\mathbf{U}$ to qubits with indices 
 $\{i+ln: 0\le l\le m-1\}$ among $mn$ qubits. Assume that we have an $[[n,k]]$ stabilizer code $\mathcal{Q}$ and that $\mathbf{E}$ is an encoding operator of $\mathcal{Q}$. $\mathcal{Q}$ is $\mathbf{U}$-transversal if:
\begin{equation}
\left(\prod_{i=1}^n \mathbf{\hat {U}}_i\right)   \bigotimes_{j=1}^{m} \left(\mathbf{E} |\bm{\psi}_j\rangle |0\ldots 0\rangle\right) =
   \mathbf{E}^{\otimes \,m}\left(\prod_{i=1}^{k}  \mathbf{\hat {U}}_i\right)  \bigotimes_{j=1}^{m} \left( |\bm{\psi}_j\rangle |0\ldots 0\rangle\right),
  \label{Eq.TrDef}
\end{equation}
where $|\bm{\psi}_j\rangle \in {\mathbb C}^{2^k}$ are logical states before encoding, and the $|0\ldots 0\rangle$ state corresponds to the $n- k$ physical qubits initialized in the $\lvert0\rangle$ state for the encoding process.

Alternatively, we may define the permutation operator $\boldsymbol{\Pi}$ as the operator acting on tensor product states as:  
\begin{equation}
\boldsymbol{\Pi}  \,  \otimes_{j=1}^m \,\lvert \psi_{j1}, \psi_{j2}, \ldots \psi_{jn} \rangle
 = 
\otimes_{j=1}^n \lvert \, \psi_{1j}, \psi_{2j}, \ldots \psi_{mj} \rangle\,.
\label{Eq.Permute}
\end{equation}
Then \eqref{Eq.TrDef} becomes:
\begin{equation}
 \boldsymbol{\Pi}^{\mathrm T} \, 
\mathbf{U}^{\otimes\, n} \, \boldsymbol{\Pi} \,  \bigotimes_{j=1}^{m} \left(\mathbf{E} |\bm{\psi}_j\rangle |0\ldots 0\rangle\right) =
  \mathbf{E}^{\otimes\, m}\,
   \boldsymbol{\Pi}^{\mathrm T}   
  \left(\mathbf{U}^{\otimes\, k}\otimes \mathbf{I}_{\,2^{(n-k)}}\right)
   \boldsymbol{\Pi} \, \bigotimes_{j=1}^{m} \left( |\bm{\psi}_j\rangle |0\ldots 0\rangle\right)\,,
  \label{Eq.TrDefNew}
\end{equation}
where $\mathbf{I}_j$ is the identity in $j$ dimensions.
In this form we clearly see how the physical transversal operation $ 
\mathbf{U}^{\otimes\, n}$ realizes the $k$-fold logical operation $\mathbf{U}^{\otimes\, k}$.


Different stabilizer codes have different sets of transversal gates, e.g., any CSS code is CNOT gate transversal. Furthermore, any symmetric CSS code has transversality for all Clifford gates.
For non-Clifford gates,  a very useful non-Clifford gate is the $\mathbf{\rm T}$-gate:
$$\mathbf{ T}=\left[\begin{array}{cc}
     1& 0 \\
     0& e^{i\pi/4}
\end{array}\right].$$ 
Certain non-symmetric CSS codes belonging to the family of triorthogonal codes are T-transversal~\cite{bravyi2005magicstate}.

In~\cite{bayanifar2025transversality}, we 
extended the notion of  transversality to hold across two distinct codes. Expanding on this, the  standard definition~\eqref{Eq.TrDef} of 
transversality within the same code can be generalized to an ordered set of $m$  $[[n,k]]$ stabilizer codes $\mathcal{Q}_j,\ j=1,\ldots,m$ with encoding operators 
$\mathbf{E}_{\rm j}$ as follows. The set of codes  $\mathcal{Q}_1,\ldots,\mathcal{Q}_m$ are $\mathbf{U}$-transversal if:
\begin{equation}
\left(\prod_{i=1}^n \mathbf{\hat {U}}_i\right)   \bigotimes_{j=1}^{m} \left(\mathbf{E}_{\rm j} |\bm{\psi}_j\rangle |0\ldots 0\rangle\right) =
  \left(\bigotimes_{s=1}^{m}  \mathbf{E}_{\rm s}\right)\left(\prod_{i=1}^{k}  \mathbf{\hat {U}}_i\right)  \bigotimes_{j=1}^{m} \left( |\bm{\psi}_j\rangle |0\ldots 0\rangle\right).
\label{Eq.TrDefm}
\end{equation}
With using the permutation operators as~\eqref{Eq.Permute}, this definition can write into the similar form as in~\eqref{Eq.TrDefNew}.
In other words, the standard definition~\eqref{Eq.TrDef} assumes that $\mathbf{U}$ operates between $m$ code states of the same code $\mathcal{Q}$, while~\eqref{Eq.TrDefm} assumes that $\mathbf{U}$ operates between $m$ code states of $m$ possibly different codes, $\mathcal{Q}_1,\ldots,\mathcal{Q}_m$. We show in the following sections that this opens up a number of new opportunities and can be very useful for fault-tolerant computations.

 In \cite{bayanifar2025transversality} we considered the case of $m=2$ in detail and formulated  the necessary and sufficient conditions for code pairs to be CNOT and CZ transversal. 
The conditions for realizing a transversal logical CNOT from
$\rm{CSS} (\mathcal{C}_1,\mathcal{C}_2)$ to  $\rm{CSS}(\mathcal{C}_3,\mathcal{C}_4)$ are:
\begin{equation}
    \mathcal{C}_1 / \mathcal{C}_2^{\perp} \cong \mathcal{C}_3 / \mathcal{C}_4^{\perp}, \quad  \mathcal{C}_2^\perp \subseteq \mathcal{C}_4^\perp.
    \label{Eq:CNOTCondition}
\end{equation}
Similarly, we have CZ gate transversality between $\rm{CSS} (\mathcal{C}_1,\mathcal{C}_2)$ and  $\rm{CSS}(\mathcal{C}_3,\mathcal{C}_4)$ iff:
\begin{align}
        \mathbf{x}^A \mathbf{z}^{\rm T} + \mathbf{y}\left( \mathbf{x}^B + \mathbf{z}\right)^{\rm T}  = 0 \qquad \forall \, &\mathbf{x}^A\in\mathcal{C}_1/\mathcal{C}_2^\perp,\,\mathbf{x}^B\in\mathcal{C}_3/\mathcal{C}_4^\perp,\,\mathbf{y}\in \mathcal{C}_2^{\perp}, \, \mathbf{z}\in \mathcal{C}_4^{\perp} \\
    &\mathbf{A B}^{\rm T} = \mathbf{I},
    \end{align}
    where 
$\mathbf{A}, \mathbf{B} \in \mathbb{F}_2^{k \times n}$,
$\mathbf{A}\cong\mathcal{C}_1 / \mathcal{C}_2^{\perp}$, $\mathbf{B}\cong\mathcal{C}_3 / \mathcal{C}_4^{\perp}$ are two mapping matrices as in~\eqref{eq:mappingMat}. 
Often, it is more convenient to use the following sufficient conditions:
\begin{equation} \label{Eq:CZsuffcond1}
        \mathcal{C}_1 / \mathcal{C}_2^{\perp} \subset \mathcal{C}_4, \quad  \mathcal{C}_3 \subseteq \mathcal{C}_2, \quad 
        \mathbf{A B}^{\rm T} = \mathbf{I}.
    \end{equation}

The concept of transversality can be extended to codes of unequal lengths.
A logical two-qubit gate between an $n$-qubit code $\mathcal{Q}_{\mathrm{A}}$ and an $n'$-qubit code $\mathcal{Q}_{\mathrm{B}}$ ($n' \le n$) is \emph{$n'$-transversal} if it decomposes into $n'$ bitwise physical two-qubit gates coupling $\mathcal{Q}_{\mathrm{B}}$ to an $n'$-qubit subset of $\mathcal{Q}_{\mathrm{A}}$, leaving the remaining $n - n'$ qubits uncoupled.
Without loss of generality, we may index the physical qubits such that the operation couples $\mathcal{Q}_{\mathrm{B}}$ directly to the first $n'$ qubits of $\mathcal{Q}_{\mathrm{A}}$.

    
\subsection{Triorthogonal Codes}

Triorthogonal codes are a special type of CSS codes, which are generated by a  triorthogonal matrix. 
   \bc{ A matrix $\mathbf{G}=[G_{ij}] \in \mathbb{F}_2^{m \times n}$ is called {\em triorthogonal} if 
    \begin{equation}
            \sum_{i=1}^n G_{ai}G_{bi}=0\!\!\!\!\mod 2 \text{, for any $a \not = b$}\,,
            \label{Eq:Cond1}
    \end{equation}   
    \begin{equation}
            \sum_{i=1}^n G_{ai}G_{bi}G_{ci}=0\!\!\!\!\mod 2\text{, for any distinct $a,b,c$\,}
    \end{equation}}
see~\cite{shi2024triorthogonal}.
By row permutation, we get 
 \begin{equation}
     \mathbf{G}=\left [
\begin{array}{c}
\mathbf{G}_1 \\
\hline
\mathbf{G}_0
\end{array}
\right],
 \label{Eq:TriMatrix}
 \end{equation}
where $\mathbf{G}_1$ and $\mathbf{G}_0$ formed by all odd and even weight rows of $\mathbf{G}$, respectively. 

\begin{definition} \label{Def:triOrthoCodes}
\bc{An $m\times n$ triorthogonal matrix $\mathbf{G}$ of the form \eqref{Eq:TriMatrix} defines an $[[n,k,d]]$ triorthogonal ${\rm CSS}\left( \mathcal{C}_1,  \mathcal{C}_2\right)$ code with stabilizers}
\begin{equation}
    \mathbf{G}_{\rm T} = 
    \left[
    \begin{array}{cc}
        \mathbf{G}_0 & \mathbf{0} \\
        \hline
        \mathbf{0} & \mathbf{G}^{\perp} 
    \end{array}
    \right]\,,
\end{equation}
\bc{where $k$ is the  number of rows of $\mathbf{G}_1$, while  $\mathbf{G}$ is the generator matrix of $\mathcal{C}_1$ and $\mathbf{G}_0^\perp$ of $\mathcal{C}_2$. There are $m-k$ X-stabilizers and  with $n-m$ Z-stabilizers, and the minimum distance $d$ is  \eqref{eq:mindist}. }
\end{definition}

\bc{It is shown in~\cite{narayana2020optimality} that if a triorthogonal CSS code $\mathcal{Q}^{\rm T}$ is Pauli $\mathbf{X}$-transversal then $\mathcal{Q}^{\rm T}$ is also $\mathbf{\rm T}$-transversal. We shall call such codes {\em T-triorthogonal}.}

\subsection{Quantum Reed-Muller Codes}

Quantum Reed-Muller (QRM) codes are a special type of CSS codes, which consist of two classical Reed-Muller (RM) codes. 
Classical RM codes are a family of linear error-correcting codes based on the evaluation of multivariate Boolean polynomials. 
A classical RM code $\mathcal{RM}(r,m)$, of order $r$ and length $n=2^m$,  is explicitly defined by its $k \times n$ -dimensioal generator matrix $ \bG(r,m)$, where $k = \sum_{i=0}^r \binom{m}{i}$. The generator matrix is constructed using the recursive Plotkin construction:
\begin{equation}
    \bG(r,m) = \begin{pmatrix}
        \bG(r, m-1) & \bG(r, m-1) \\
        \mathbf{0} & \bG(r-1, m-1)
    \end{pmatrix},
\end{equation}
with the boundary conditions:
\begin{align}
    \bG(0, m) & = \mathbf{1}_{1 \times 2^m} , \\
    \bG(m, m) &= \mathbf{I}_{2^m}\,.
\end{align}
Here $\mathbf{1}_{1 \times 2^m}$ denotes a length $2^m$ all-ones vector, and $\mathbf{I}_{2^m}$ represents the $2^m \times 2^m$ identity matrix.
An important feature of classical RM codes is that:
\begin{equation}
    \mathcal{RM}(0,m)\subset\mathcal{RM}(1,m)\subset\dots\subset\mathcal{RM}(m,m).
    \label{Eq.RMSubset}
\end{equation}
The dual code of $\mathcal{RM}(r,m)$ is given by:
\begin{equation}
    \mathcal{RM}(r,m)^\perp = \mathcal{RM}(m - r - 1, m).
\label{Eq.DualRMCode}
\end{equation}
Consequently, the parity check matrix of the code $\mathcal{RM}(r,m)$ is the generator matrix $\bG(m - r - 1, m)$ of its dual code.

\subsubsection{Quantum Reed-Muller Codes with Length $n=2^m$}

QRM codes are constructed by embedding these classical codes into the CSS framework.
As discussed in the  Section~\ref{Sec.Pre}, a CSS code requires two classical linear codes, $\mathcal{C}_{\rm X}=\mathcal{C}_2$ and $\mathcal{C}_{\rm Z}=\mathcal{C}_1$, such that $\mathcal{C}_{\rm X}^\perp \subseteq \mathcal{C}_{\rm Z}$. By selecting $\mathcal{C}_{\rm X}^\perp = \mathcal{RM}(r_x, m)$ and $\mathcal{C}_{\rm Z}^\perp = \mathcal{RM}(r_z, m)$, we thus have to have  $\mathcal{RM}(r_x, m) \subseteq \mathcal{RM}(m-r_z-1, m)$. This  holds if and only if:
\begin{equation}
    r_x + r_z < m.
\end{equation}
When this condition holds we have a CSS code $\text{QRM}(r_x,r_z,m)$ which encodes 
\begin{equation}
k=\left\{\begin{array}{lc}
     \sum_{l=r_z+1}^{m-r_x-1} \binom{m}{l},& \quad r_x+r_z<m-1 \\
     0, & \quad r_x+r_z=m-1
\end{array}\right.    
\end{equation}
logical qubits into $n = 2^m$ physical qubits.  

\subsubsection{Punctured Quantum Reed-Muller Codes}

While the standard QRM codes of the previous section have parameters $[[n = 2^m,k]]$, puncturing one physical qubit and two stabilizer generators leads to $[[2^m-1,k+1]]$ PQRM codes~\cite{ross2016optimalancillafreecliffordtapproximation,gong2024computationquantumreedmullercodes,anderson2014codeswitch}, 
with higher rate and less physical qubits.

For this, shortened classical RM  generator matrices $\bG_{\rm s}(r,m)$ are generated by puncturing one row and one column from $\bG(r,m)$, such that their intersection is 1. Without loss of generality, we remove the first row and column.  A stabilizer 
generator matrix of ${\rm PQRM}(r_x,r_z,m)$ code can then be represented as:
\begin{equation}
    \mathbf{G}_{\rm{PQRM}}=\begin{bmatrix}
        \bG_{\rm s}(r_x,m)&\mathbf{0}\\
        \mathbf{0}&\bG_{\rm s}(r_z,m)
    \end{bmatrix}.
\end{equation}
We can correspondingly define a punctured classical code $\mathcal{C}_{\rm p}$ as  $\mathcal{C}$ with the first bit removed, generated  
by the punctured RM matrix
$\bG_{\rm p}$.

Using the condition of the dual code of the shortened code~\cite{macwilliams1977theory}: 
\begin{equation}
    (\mathcal{C}_{\rm s})^\perp=(\mathcal{C^\perp})_{\rm p}
    \label{Eq.shortenDual}
\end{equation}
and combining it with~\eqref{Eq.DualRMCode}, we get the dual code of a shortened RM code as:
\begin{equation}
    \mathcal{RM}^\perp_{\rm s}(r,m)=\mathcal{RM}_{\rm p}(m-r-1,m)\,.
    \label{Eq.DualofRM_s}
\end{equation}
Accordingly, a one-bit punctured PQRM code 
encodes $k=\sum_{l=r_z+1}^{m-r_x-1} \binom{m}{l}+1$ logical qubits, and it has a code distance $d=2^{\min (r_x,r_z)+1}-1$~\cite{Hastings_2018}.


Note that a shortened RM-matrix can be decomposed into lower-rank RM matrices as:
\begin{align}
    &\bG_{\rm s}(r,m)=\begin{bmatrix}
            \bG_{\rm s}(r,m-1)& \bG_*(r,m-1)\\
            \mathbf{0} & \bG(r-1,m-1)
        \end{bmatrix}.
        \label{Eq.RMDecomp}
\end{align}
Here, the lower script $*$ denotes that the first row of the matrix has been removed.

In this paper, we focus on $k=1$, for which $r_z=m-r_x-1$.
Such PQRM codes admit transversal gates in level $v$ of the Clifford hierarchy if $r\leq\frac{m-1}{v}$~\cite{gong2024computationquantumreedmullercodes,Hastings_2018}. 
Many triorthogonal codes can be described as PQRM codes, with transversal gates in $v\geq3$ level Clifford hierarchy, examples can be found in~\cite{nezami2022classification}.
We denote the $k=1$ PQRM codes considered here as:
\begin{equation}
    {\cal Q}_m^r \equiv 
    {\rm PQRM}(r, m-r-1,m)
    \equiv [[2^m-1,1,2^{\min (m-r-1,r)+1}-1]].
    \label{Eq.CSSDefPQRM}
\end{equation}

\section{Transversal Code Switching}
\label{Sec.CodeGeneration}

According to the Eastin-Knill theorem, no QECC can support a universal transversal gate set~\cite{eastin2009restrictions}. However, this  does not prevent one from using multiple codes within a system to form a universal gate set. Specifically, if one code admits transversal non-Clifford gates and another supports transversal Clifford gates, 
then by teleporting logical information between them,
a universal gate set can be implemented using only transversal operations within each code~\cite{anderson2014codeswitch,vuillot2019code}.




Several studies have explored fault-tolerant code switching protocols between the Steane code and QRM codes~\cite{anderson2014codeswitch,vuillot2019code}. In these frameworks, the switching process is typically formulated via gauge fixing within a subsystem code. By performing sequential stabilizer measurements combined with targeted feedback operations, information encoded in one code is mapped onto the other code. 

Building on our earlier works~\cite{bayanifar2025transversality,jiao2023using,jiao2025quantum}, 
where one-way transversal CNOTs and CZs were used to generate long-distance entanglement between two codes, we teleport logical states between two code blocks encoded by $\mathcal{Q}^{\rm T}$ and $\mathcal{Q}^{\rm Sym}$ using transversal state teleportation. 
This enables the desired logical operations to be implemented transversally throughout the process. 
We refer to this method as {\em transversal code switching}.  
A related method was proposed in~\cite{Heu_en_2025TransversalCodeSwithcing}, where the authors perform transversal code switching between 2D and 3D color codes using one-way transversal CNOT gates.

Given an $[[n,k,d]]$ CSS code $\mathcal{Q}^{\rm T}$ that admits a transversal $\mathbf{\rm T}$ gate, the goal is to construct a paired  $[[n',k,d']]$ symmetric CSS code $\mathcal{Q}^{\rm Sym}$ with $n'\leq n$, such that these two codes satisfy the sufficiency conditions~\eqref{Eq:CZsuffcond1} or~\eqref{Eq:CNOTCondition}, i.e., being CNOT or CZ-transversal. For the specific case of PQRM code, we provide a more general result related to  CNOT transversality between such codes. 

\subsection{One-way Transversal CNOT Gate}
\label{Sec.one-wayCNOTgeneral}

Let $\mathcal{Q}^T = \mathrm{CSS}(\mathcal{C}_1, \mathcal{C}_2)$ be an $[[n,k]]$ triorthogonal code with generator matrices $\mathbf{G}_1$ and $\mathbf{G}_2$. Similarly, let $\mathcal{Q}^{\mathrm{sym}} = \mathrm{CSS}(\mathcal{C}_{\mathrm{sym}}, \mathcal{C}_{\mathrm{sym}})$ be a symmetric code whose classical components are generated by $\mathbf{G}_{\mathrm{sym}}$.
Because transversal Hadamard operations on the symmetric code convert a one-way transversal CNOT directed from $\mathcal{Q}^T$ to $\mathcal{Q}^{\mathrm{sym}}$ into a transversal CZ gate, the transversality requirements in~\eqref{Eq:CNOTCondition} and~\eqref{Eq:CZsuffcond1} coincide.
Consequently, it suffices to analyze the CNOT conditions in~\eqref{Eq:CNOTCondition}. However, the relations in~\eqref{Eq:CNOTCondition} require both codes to have the same number of physical qubits. To apply these conditions directly, we construct an extended symmetric code $\mathcal{Q}'^{\mathrm{\rm sym}} \equiv \mathrm{CSS}(\mathcal{C}_3, \mathcal{C}_4)$, where $\mathcal{C}_3$ and $\mathcal{C}_4$ are defined via the generator matrices:
\begin{equation}
    \mathbf{G}_3=\begin{bmatrix}
    \mathbf{G}_{\rm sym}&\mathbf{0}\\
    \mathbf{0}&\mathbf{I}_m
    \end{bmatrix},
    \quad\mathbf{G}_4=\begin{bmatrix}
        \mathbf{G}_{\rm sym}&\mathbf{0}
    \end{bmatrix},
\end{equation}
with $m = n - n'$. From this construction, the dual generator matrix can be expressed as:
\begin{equation}
\label{Eq.ExtendedCode}
  \mathbf{G}_3^\perp=\begin{bmatrix}
    \mathbf{G}_{\rm sym}^\perp&\mathbf{0}
    \end{bmatrix},\quad  \mathbf{G}_4^\perp=\begin{bmatrix}
        \mathbf{G}_{\rm sym}^\perp&\mathbf{0}\\
        \mathbf{0}&\mathbf{I}_m
    \end{bmatrix}.
\end{equation}
With this modified code, the transversality conditions in~\eqref{Eq:CNOTCondition} become applicable. As shown later, these auxiliary appended qubits introduce only trivial stabilizer extensions and can therefore be eliminated in practical implementations. 
From~\eqref{Eq:CNOTCondition}, the transversal CNOT gate requires that
$\mathcal{C}_2^\perp\subset\mathcal{C}_4^\perp$, 
which is equivalent to $\mathcal{C}_4 \subset \mathcal{C}_2$, the generator matrix $\mathbf{G}_2$ can be partitioned as:
\begin{equation}
    \mathbf{G}_2=\begin{bmatrix}
        \mathbf{G}_4\\
        \mathbf{B}
    \end{bmatrix},
\end{equation}
and $\mathbf{G}_4=\begin{bmatrix}
    \mathbf{G}_{\rm sym}&\mathbf{0}
\end{bmatrix}$. Orthogonality enforces that, without loss of generality, $\mathbf{G}_2^\perp$ admits the block upper-triangular representation:
\begin{equation}
    \mathbf{G}_2^\perp=\begin{bmatrix}
        \mathbf{G}_{\rm sym}^\perp&\mathbf{B}_1'\\
        \mathbf{0}&\mathbf{B}_2'
    \end{bmatrix},
\end{equation}
which $\mathbf{B}_1'$ and $\mathbf{B}_2'$ are matrices related to $\mathbf{B}$. 
Furthermore, using the quotient space canonical isomorphism $\mathcal{C}_1 / \mathcal{C}_2^\perp \cong \mathcal{C}_3 / \mathcal{C}_4^\perp$, we can choose row-equivalent generator matrices of the form:
\begin{equation}
 \mathbf{G}_1=\begin{bmatrix}
     \mathbf{G}_2^\perp\\
     \mathbf{A}
 \end{bmatrix},\quad \mathbf{G}_3=\begin{bmatrix}
     \mathbf{G}_4^\perp\\
     \mathbf{A}
 \end{bmatrix} . 
\end{equation}
This leads to a sufficient condition $\mathcal{C}_3^\perp\subset\mathcal{C}_1^\perp$, by which we have:
\begin{equation}
    \mathbf{G}_1^\perp=\begin{bmatrix}
        \mathbf{G}_3^\perp\\
        \mathbf{D}
    \end{bmatrix}=\begin{bmatrix}
        \mathbf{G}_{\rm sym}^\perp&\mathbf{0}\\
        \mathbf{D}_1&\mathbf{D}_2
    \end{bmatrix},
\end{equation}
where $\mathbf{D}=\begin{bmatrix}
    \mathbf{D}_1&\mathbf{D}_2
\end{bmatrix}$, and the dimension of $\mathbf{G}_{\rm sym}^\perp$ is $\frac{n-m-k}{2}\times n-m$.

We conclude that when $\mathbf{G}_1^\perp$ and $\mathbf{G}_2^\perp$ admit such upper- and lower-block-diagonal decompositions, the $\mathcal{Q}^{\rm Sym}$ symmetric CSS code with fewer physical qubits can share the one-way transversal CNOT gate in between. Specifically, $n$ parallel CNOT gates from $\mathcal{Q}^{\rm T}$ to $\mathcal{Q}^{\rm Sym}$ are equivariant to a logical CNOT gate.  To realize this one-way operation, we apply $n$ transversal CNOT gates in parallel between the triorthogonal code and the extended symmetric code. The first $n-m$ gates act between the data qubits of the triorthogonal code and original symmetric code, whereas the remaining $m$ gates target the appended physical ancilla qubits from data qubits in $\mathcal{Q}^{\rm T}$. By the construction of the extended CSS code generator in~\eqref{Eq.ExtendedCode}, all $m$ appended qubits are initialized in the $\lvert +\rangle$ state (the $+1$ eigenstate of the Pauli $X$ operator). A CNOT gate targeting a state initialized in $\lvert +\rangle$ acts trivially:
\begin{equation}
\mathrm{CNOT} \left( \lvert\psi\rangle \otimes \lvert+\rangle \right) = \lvert\psi\rangle \otimes \lvert+\rangle.
\end{equation}
Consequently, the $m$ CNOT gates acting on the appended ancilla qubits can be omitted entirely. 
Because the generator matrix confirms that these auxiliary qubits share no entanglement with the original symmetric code block, they can be safely discarded from the circuit.

\subsection{Transversal Code Switching Between PQRM Codes}

We begin by focusing on the PQRM codes, which share an intersection with the family of triorthogonal codes. Within this family, certain instances directly correspond to color codes in different spatial dimensions; for example, $\mathrm{PQRM}(1,1,3)$, $\mathrm{PQRM}(1,2,4)$, and $\mathrm{PQRM}(1,3,5)$ map to 2D, 3D, and 4D color codes, respectively. While transversal code switching between 2D and 3D color codes, specifically the 7-qubit Steane code and the 15-qubit tetrahedral code, was established in~\cite{Heu_en_2025TransversalCodeSwithcing}, we generalize this mechanism to color codes in arbitrary dimensions. Furthermore, we introduce a systematic \emph{code map} to identify switchable code pairs for any given PQRM code.
Using our cross-code transversality framework, we prove that these transversal CNOT gates implement logical CNOT gates between different code blocks. As an example, we show that the CNOT gate from the $\mathcal{Q}_5^1\equiv[[31,1,3]]$ PQRM code to the $\mathcal{Q}_3^1\equiv[[7,1,3]]$ Steane code is transversal, which corresponds to switching from a 4D color code to a 2D color code~\cite{Brown_2016}.  
This enables the transversal $\sqrt{\mathbf{T}}$ gate, which offers a finer-angle rotation, significantly reducing circuit depth in algorithms such as the QFT algorithm.
Moreover, it can be integrated into the Steane error correction procedure without incurring additional overhead. We will demonstrate how this can be done in Sec.~\ref{sec.Steaneframe}.

We can completely characterize when there is a transversal CNOT between two PQRM codes. First, for the existence of an one-way transversal CNOT gate between two PQRM codes, we have:

\begin{lemma}
    For the PQRM codes $\mathcal{Q}_m^r$ and $\mathcal{Q}_{m-1}^r$, the $2^{m-1}-1$ -transversal physical CNOT gate from $\mathcal{Q}_m^r$ to $\mathcal{Q}_{m-1}^r$ is equivalent to the logical CNOT gate.
    \label{Lemma.m-to-m-1}
\end{lemma}
\begin{proof}
    Let $\mathcal{Q}_m^r \equiv {\rm CSS}(\mathcal{C}_{\rm Z}, \mathcal{C}_{\rm X})$. Based on the recursive construction of RM codes, its $X$- and $Z$-stabilizer generator matrices can be written as in~\eqref{Eq.RMDecomp}:
    \begin{align}  \label{Eq:HZdef}
        &\mathbf{H}_{\rm X}(m)=\bG_{\rm s}(r,m)=\begin{bmatrix}
            \bG_{\rm s}(r,m-1)& \bG_*(r,m-1)\\
            \mathbf{0} & \bG(r-1,m-1)
        \end{bmatrix},\nonumber\\ &\mathbf{H}_{\rm Z}(m)=\bG_s(m-r-1,m)=\begin{bmatrix}
            \bG_{\rm s}(m-r-1,m-1)& \bG_*(m-r-1,m-1)\\
            \mathbf{0}& \bG(m-r-2,m-1)
        \end{bmatrix}.
    \end{align}

    Similarly, for the smaller code $\mathcal{Q}_{m-1}^r$, the corresponding $X$- and $Z$-stabilizer generator matrices are:
    \begin{equation}
        \mathbf{H}_{\rm X}(m-1)=\bG_{\rm s}(r,m-1),\quad\mathbf{H}_{\rm Z}(m-1)=\bG_{\rm s}(m-r-2,m-1).
    \end{equation}
  
    We construct an extended code $\mathcal{Q}'^r_{m-1} \equiv {\rm CSS}(\mathcal{C}_{\rm Z}', \mathcal{C}_{\rm X}')$ using the same structure in~\eqref{Eq.ExtendedCode}:
    \begin{align}
        &\mathbf{H}_{\rm X}'(m-1)=\begin{bmatrix}
            \bG_{\rm s}(r,m-1)&\mathbf{0}\\
            \mathbf{0}&\mathbf{I}_{2^{m-1}}
        \end{bmatrix}\nonumber\\
        &\mathbf{H}_{\rm Z}'(m-1)=\left [\bG_{\rm s}(m-r-2,m-1), \mathbf{0}\right].
    \end{align}

With these extended generators, we satisfy the first CSS condition for transversal CNOT gates: $\mathcal{C}_{\rm X}^\perp \subseteq \mathcal{C}_{\rm X}'^\perp$. 

Regarding the generator matrix of the classical code $\mathcal{C}_{\rm Z}$, we use the duality condition~\eqref{Eq.DualofRM_s}. We can partition the generator matrix of $\mathcal{C}_{\rm Z}$ into the stabilizers and the logical operators:
\begin{equation}
\label{Eq.Garb}
    \mathbf{G}_{\rm Z} = \bG_{\rm p}(r,m)= \begin{bmatrix}
        \mathbf{H}_{\rm X}(m) \\
        \mathbf{A}
    \end{bmatrix},
\end{equation}
where $\mathbf{A} = \left[ \mathbf{g}^{\rm p}_1(r,m-1),\mathbf{g}_1(r,m-1) \right]$ with $\mathbf{g}^{\rm p}_i(r,m)$ and $\mathbf{g}_i(r,m)$ denoting the $i$-th row  of matrices $\bG(r,m)$ and $\bG_{\rm p}(r,m)$.
Consequently, the quotient space is spanned by the rows of this matrix: $\mathcal{C}_{\rm Z} \slash \mathcal{C}_{\rm X}^\perp \cong \operatorname{span}(\mathbf{A})$.

For the constructed extended code $\mathcal{Q}'^r_{m-1}$, we similarly obtain the generator matrix of $\mathcal{C}_{\rm Z}'$ as:
\begin{equation}
    \mathbf{G}_{\rm Z}' = \begin{bmatrix}
        \bG_{\rm p}(r,m-1)&\mathbf{0}\\
        \mathbf{0}& \mathbf{I}
    \end{bmatrix}= \begin{bmatrix}
        \mathbf{H}_{\rm X}'(m-1) \\
        \mathbf{A}'
    \end{bmatrix},
\end{equation}
where $\mathbf{A}'=[\mathbf{g}^{\rm p}_1(r,m-1),\mathbf{0}]$.
It follows that the quotient space $\mathcal{C}_{\rm Z}' / \mathcal{C}_{\rm X}'^{\perp}$ is spanned by the rows of $\mathbf{A}'$.
Note that the choice of coset representatives for $\mathbf{A}'$ is not unique, because the rows of $\mathbf{H}_{\rm X}'(m-1)$ generate $\mathcal{C}_{\rm X}'^{\perp}$, adding any linear combination of these stabilizer rows to $\mathbf{A}'$ yields an equivalent set of representatives. By this equivalence, $\mathbf{A}$ transforms into an equally valid selection for $\mathbf{A}'$. From this, it follows that the rows of $\mathbf{A}$ serve as valid coset representatives for both $\mathcal{C}_{\rm Z}/\mathcal{C}_{\rm X}^\perp$ and $\mathcal{C}_{\rm Z}'/\mathcal{C}_{\rm X}'^{\perp}$. Thus, we conclude that the logical spaces are isomorphic, $\mathcal{C}_{\rm Z}/\mathcal{C}_{\rm X}^\perp \cong \mathcal{C}_{\rm Z}'/\mathcal{C}_{\rm X}'^{\perp}$.

Having established the one-way transversal CNOT operation between the PQRM code $\mathcal{Q}_m^r$ and the extended code ${\mathcal{Q}'}_{m-1}^r$, following the discussion in Section.~\ref{Sec.one-wayCNOTgeneral}, the additional $2^{m-1}$ CNOT gates targeting the appended physical qubits can be omitted. Consequently, this leaves exactly $2^{m-1}-1$ parallel physical CNOT gates acting between $\mathcal{Q}_m^r$ and $\mathcal{Q}_{m-1}^r$ to fully implement the transversal logical CNOT gate.
\end{proof}

Besides reducing $m$, we can also increase $r$ to find another symmetric code.

\begin{lemma}
    For PQRM codes $\mathcal{Q}^r_m$ and $\mathcal{Q}^{r+1}_m$, the $2^m-1$ -transversal  physical CNOT gate from $\mathcal{Q}^r$ to $\mathcal{Q}^{r+1}$ is equivalent to the logical CNOT gate.
    \label{Lemma.r-to-r+1}
\end{lemma}
\begin{proof}
    Let $\mathcal{Q}_m^{r} \equiv {\rm CSS}(\mathcal{C}_{\rm Z}, \mathcal{C}_{\rm X})$. The parity-check matrices of $\mathcal{C}_{\rm Z}$ and $\mathcal{C}_{\rm X}$ are given by:
    \begin{equation}
        \mathbf{H}_{\rm X}(r)=\bG_{\rm s}(r,m), \quad \mathbf{H}_{\rm Z}(r)=\bG_{\rm s}(m-r-1,m).
    \end{equation}

    Similarly, for the target code $\mathcal{Q}_m^{r+1} \equiv {\rm CSS}(\mathcal{C}_{\rm Z}', \mathcal{C}_{\rm X}')$, the corresponding parity-check matrices are:
    \begin{equation}
        \mathbf{H}_{\rm X}(r+1)=\bG_{\rm s}(r+1,m), \quad \mathbf{H}_{\rm Z}(r+1)=\bG_{\rm s}(m-r-2,m).
    \end{equation}
    Following the standard classical RM code inclusion condition~\eqref{Eq.RMSubset}, it follows that $\mathcal{C}_{\rm X}^\perp(r) \subseteq \mathcal{C}_{\rm X}'^\perp(r+1)$. This satisfies the first condition for transversal CNOT gates between CSS codes.
    The structure of the logical spaces, represented by the quotient space $\mathcal{C}_{\rm Z} / \mathcal{C}_{\rm X}^\perp$, can be analyzed by explicitly constructing the generator matrices for the $Z$-codes, as in \eqref{Eq.Garb}  in the proof of Lemma~\ref{Lemma.m-to-m-1}. We partition these generator matrices into their stabilizer and logical components:
    \begin{align}
        \mathbf{G}_Z(r) &= \bG_{\rm p}(r,m) = \begin{bmatrix}
            \mathbf{H}_{\rm X}(r) \\
            \mathbf{A}
        \end{bmatrix}, \label{Eq.PQRM(r)} \\
        \mathbf{G}_{\rm Z}'(r+1) &= \bG_{\rm p}(r+1,m) = \begin{bmatrix}
            \mathbf{H}_{\rm X}(r+1) \\
            \mathbf{A}'
        \end{bmatrix}, \label{Eq.PPQRM(r+1)}
    \end{align}
   where the logical operators are defined by the punctured first rows of the respective classical matrices: $\mathbf{A} = \mathbf{g}^{\rm p}_1(r,m)$ and $\mathbf{A}' = \mathbf{g}^{\rm p}_1(r+1,m)$. 
    
    Because the classical RM codes are strictly nested such that $\mathcal{RM}(r,m) \subseteq \mathcal{RM}(r+1,m)$, the vector $\mathbf{A}$, which belongs to $\bG_{\rm p}(r,m)$, must also be an element of $\bG_{\rm p}(r+1,m)$. 
    The space of $\bG_{\rm p}(r+1,m)$ is entirely spanned by the rows of $\mathbf{H}_{\rm X}(r+1)$ and $\mathbf{A}'$. Therefore, $\mathbf{A}$ can be expressed as $\mathbf{A}'$ plus some linear combination of the rows from the stabilizer matrix $\mathbf{H}_{\rm X}(r+1)$. 
    Consequently, we can find a linear combination such that $\mathbf{A}$ is an equivalent selection for the coset representative $\mathbf{A}'$.
    We conclude that the logical spaces are isomorphic, $\mathcal{C}_{\rm Z}(r)/\mathcal{C}_{\rm X}^{\perp}(r) \cong \mathcal{C}_{\rm Z}'(r+1)/\mathcal{C}_{\rm X}'^{\perp}(r+1)$. Thus, the transversal physical CNOT gates perform a valid logical CNOT gate between the two codes. 
\end{proof}

Based on these lemmas, we can derive the CNOT transversal condition for such parameters:

\begin{theorem}
    For the PQRM codes $\mathcal{Q}_{m_i}^{r_i}\equiv{\rm PQRM}(r_i,m_i-r_i-1,m_i)$ and $\mathcal{Q}_{m_j}^{r_j}\equiv{\rm PQRM}(r_j,m_j-r_j-1,m_j)$ with $m_i\geq m_j$ and $r_i\leq r_j$, the logical CNOT gate is $(2^{m_j}-1)$-transversal $\mathcal{Q}_{m_i}^{r_i}$ to $\mathcal{Q}_{m_j}^{r_j}$. 
    \label{Theorem.mi-to-mj}
\end{theorem}
\begin{proof}
First, we can select $r_i=r_j=r$. 
Assuming $i > j$, we follow the idea in Lemma~\ref{Lemma.m-to-m-1} to construct the new expanded code $\mathcal{Q}_j'$. We have the generator matrices as:
\begin{align}
    &\mathbf{H}_{\rm X}'(j)=\begin{bmatrix}
        \bG_{\rm s}(r,m_j)&\mathbf{0}\\
        \mathbf{0}&\mathbf{I}_{(2^{m_i}-2^{m_j})}
    \end{bmatrix},\nonumber\\
    &\mathbf{H}_{\rm Z}'(j)=\begin{bmatrix}
        \bG_{\rm s}(m_j-r-1,m_j)&\mathbf{0}
    \end{bmatrix}.
\end{align}

Then we can find the first condition holds $\mathcal{C}_{\rm X}^\perp\subset\mathcal{C}_{\rm X}'^{\perp}$. 
For the second condition $\mathcal{C}_{\rm Z}/\mathcal{C}_{\rm X}^\perp\cong \mathcal{C}_{\rm Z}'/\mathcal{C}_{\rm X}'^{\perp}$, we know that if $\mathcal{C}_1/\mathcal{C}_2^\perp\cong\mathcal{C}_3/\mathcal{C}_4^\perp$ and $\mathcal{C}_3/\mathcal{C}_4^\perp\cong \mathcal{C}_5/\mathcal{C}_6^\perp$, then we can conclude that $\mathcal{C}_1/\mathcal{C}_2^\perp\cong \mathcal{C}_5/\mathcal{C}_6^\perp$. By using this condition, we construct an iterative chain of intermediate codes between $\mathcal{Q}_{m_i}^r$ to $\mathcal{Q}_{m_j}^r$, as $\mathcal{Q}_{m_i}\to\mathcal{Q}_{m_{i+1}}\to...\mathcal{Q}_{m_j}$. Then we can conclude that the second condition also holds.  
    
Next, we can select that $i=j$, and consider the transition between codes of different orders $\mathcal{Q}_{m_i}^{r_i}$ and $\mathcal{Q}_{m_i}^{r_j}$. 
We use Lemma~\ref{Lemma.r-to-r+1} iteratively. By stepping through intermediate codes, we can construct a valid logical path connecting $\mathcal{Q}_{m_i}^{r_i}$ and $\mathcal{Q}_{m_i}^{r_j}$. This guaranties that the required conditions for a one-way transversal CNOT gate are fulfilled between these two specific codes.

Finally, we compose the two established routes. By connecting the path ($\mathcal{Q}_{m_i}^{r_i} \to \mathcal{Q}_{m_j}^{r_i}$) and  ($\mathcal{Q}_{m_j}^{r_i} \to \mathcal{Q}_{m_j}^{r_j}$), we establish a complete operational connection between the initial code $\mathcal{Q}_{m_i}^{r_i}$ and the target code $\mathcal{Q}_{m_j}^{r_j}$. 
Because both paths strictly satisfy the one-way CNOT transversality conditions, the composite mapping between $\mathcal{Q}_{m_i}^{r_i}$ and $\mathcal{Q}_{m_j}^{r_j}$ satisfies them as well. 
\end{proof}


By this way, we can create a map to look for proper symmetric code with one-way transversal CNOT gates, as shown in Fig.~\ref{fig:PQRMTransversalMap}.
For a given PQRM code $\mathcal{Q}_m^r$, one can find a paired symmetric code $\mathcal{Q}_{2s+1}^s$ for any integer $s$ satisfying $r \le s \le \frac{m-1}{2}$. 
Codes in each column have the same code distance, and in each row, the code distance first increases from left to right, reaches the maximum at the symmetric code, and then decreases. With $m$ as an even number situation, the code distance reaches its peak at $r=\frac{m}{2}$ or $r=\frac{m}{2}-1$.
Codes in the first column have the smallest code distance $d=2^{r+1}-1$ among all the codes, while $\mathcal{Q}_m^{\frac{m-1}{2}}$ has the largest code distance of $d=2^{\frac{m+1}{2}}-1$.

\begin{figure}
    \centering
    \includegraphics[width=0.85\linewidth]{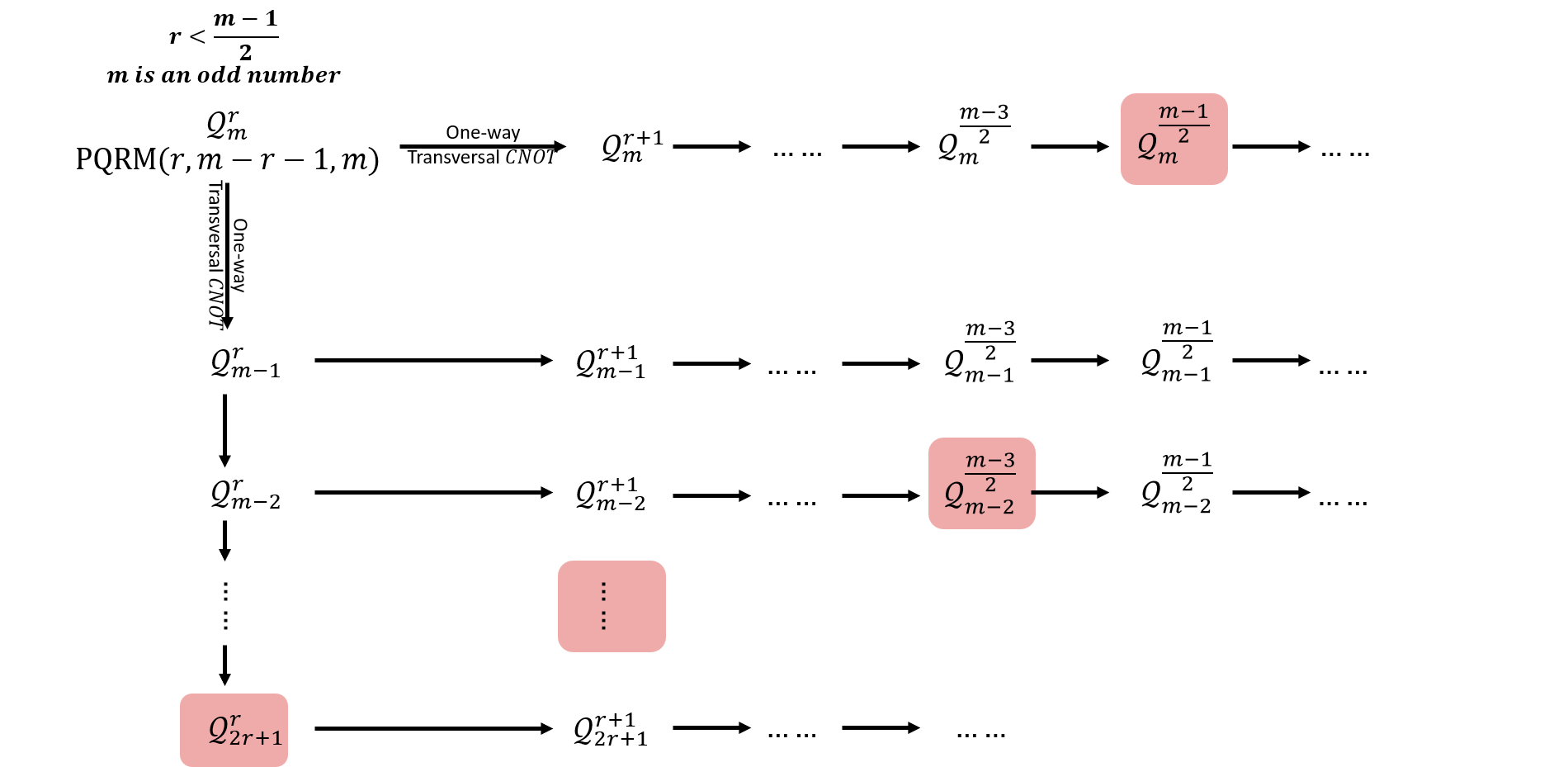}
    \caption{The PQRM code set that share one-way transversal CNOT gate with the ${\rm PQRM}(r,m-r-1,m)$ code, and $r<\frac{m-1}{2}$ and $m$ as odd number. $\mathcal{Q}_m^r$ has one-way transversal CNOT gate with all the other codes shown in this picture. Codes marked with red cube are symmetric codes. For each row, the sequence stops at when $m-r-1\leq0$.}
    \label{fig:PQRMTransversalMap}
\end{figure}

In following section, we will provide an example, to show that the CNOT gate from 4D color code, which is $[[31,1,3]]$ $\sqrt{\mathbf{\rm T}}$-transversal triorthogonal code, to 2D color code (Steane code) is transversal.

\begin{example}
\label{Ex.SteanetoQRM}
    The logical CNOT gate from the $[[31,1,3]]$ triorthogonal code to the $[[7,1,3]]$ Steane code can be achieved by a 7-transversal physical CNOT gate from the $[[31,1,3]]$ code to the Steane code.
\end{example}

    The $\mathbf{X}$ generator matrix $\mathbf{H}_{\rm X}$ of $[[31,1,3]]$ code can be represented as:
    \begin{align}
        \mathbf{G}&=\left[\begin{array}{c}
             \mathbf{H}_{\rm X}  \\
             \hline
             \mathbf{G}_1            
        \end{array}\right]=
\setlength{\arraycolsep}{1.5pt} 
\left[ \begin{array}{ccccccccccccccccccccccccccccccc}
1&0&1&0&1&0&1&0&1&0&1&0&1&0&1&0&1&0&1&0&1&0&1&0&1&0&1&0&1&0&1\\
0&1&1&0&0&1&1&0&0&1&1&0&0&1&1&0&0&1&1&0&0&1&1&0&0&1&1&0&0&1&1\\
0&0&0&1&1&1&1&0&0&0&0&1&1&1&1&0&0&0&0&1&1&1&1&0&0&0&0&1&1&1&1\\
0&0&0&0&0&0&0&1&1&1&1&1&1&1&1&0&0&0&0&0&0&0&0&1&1&1&1&1&1&1&1\\
0&0&0&0&0&0&0&0&0&0&0&0&0&0&0&1&1&1&1&1&1&1&1&1&1&1&1&1&1&1&1\\
\hline
1&1&1&1&1&1&1&1&1&1&1&1&1&1&1&1&1&1&1&1&1&1&1&1&1&1&1&1&1&1&1
\end{array}\right],\nonumber
\end{align}
where $\mathbf{G}$ is the triorthogonal matrix and $\mathbf{G}_1$ represents the logical $X$ operator of the code. The $Z$-stabilizer generator matrix $\mathbf{H}_{\rm Z}$ is obtained as the dual matrix of $\mathbf{G}$. We can partition these parity-check matrices in block form as
    \begin{equation}
        \mathbf{H}_{{\rm X}}=\left[\begin{array}{c|c}
            \mathbf{H}_{{\rm sym}} & \mathbf{C} \\
            \hline
             \mathbf{0}&\mathbf{D} 
        \end{array}\right], \quad \mathbf{H}_{{\rm Z}}=\left[\begin{array}{c|c}
            \mathbf{H}_{{\rm s}} & \mathbf{0} \\
            \hline
             \mathbf{E}&\mathbf{F} 
        \end{array}\right],
    \end{equation}
where $\mathbf{H}_{\rm sym}$ serves as both the $X$- and $Z$-stabilizer generator matrix of the self-dual Steane code. To satisfy the CNOT transversality conditions in~\eqref{Eq:CNOTCondition}, we first construct the extended Steane code $\rm CSS (\mathcal{C}_{\rm sz},\mathcal{C}_{\rm sx})$ with generator matrices as follows:
    \begin{equation}
         \mathbf{H}_{{\rm sz}}=[\mathbf{H}_{{\rm s}},\mathbf{0}_{3\times24}],\quad \mathbf{H}_{{\rm sx}}=\left[\begin{array}{cc}
        \mathbf{H}_{{\rm s}} & \mathbf{0}  \\
         \mathbf{0}& \mathbf{I}_{24} 
    \end{array}\right].
    \end{equation}
These new generators indicate that we add 24 non-entangled physical qubits to the Steane code block, and initialize them in physical $\lvert+\rangle$ states. Under these extended generator matrices, the inclusion relation $\mathcal{C}_{\rm X}^\perp \subset {\mathcal{C}_{\rm sx}}^\perp$ is immediately satisfied.For the $[[31,1,3]]$ code, the quotient space representing the logical operator can be expressed as
    \begin{equation}
    \label{Eq.31Code}
    \mathcal{C}_{\rm Z}/\mathcal{C}_{\rm X}^\perp\cong\left[\begin{array}{cc}
        \mathbf{1}_{1\times15} & \mathbf{0}_{1\times16}  
    \end{array}\right]\oplus \mathcal{C}_{\rm X}^\perp,
    \end{equation}
    and for the new Steane code, we have:
    \begin{equation}
        \label{Eq.31Steane}
        \mathcal{C}_{\rm sz}/\mathcal{C}_{\rm sx}^\perp\cong\left[\begin{array}{cc}
        \mathbf{1}_{1\times7} & \mathbf{0}_{1\times24}  
    \end{array}\right]\oplus\mathcal{C}_{\rm sx}^\perp.
    \end{equation}
    One can verify that the representative vector $\begin{bmatrix} \mathbf{1}_{1\times15} & \mathbf{0}_{1\times16} \end{bmatrix}$ is contained in both quotient representations~\eqref{Eq.31Code} and~\eqref{Eq.31Steane}. Consequently, the isomorphism condition $\mathcal{C}_{\rm Z}/\mathcal{C}_{\rm X}^\perp \cong \mathcal{C}_{\rm sz}/{\mathcal{C}_{\rm sx}}^\perp$ is strictly satisfied.

Beyond the $[[2^r-1,1,3]]$ simplex codes~\cite{Zeng_2007}, the PQRM family contains many other viable code instances, of which we provide more examples in Table~\ref{Table.Examples}. Importantly, a single triorthogonal code may admit multiple compatible symmetric partner codes characterized by different parameters. The choice of symmetric code can thus be tailored to specific hardware constraints and performance targets; e.g. $[[31,1,7]]$ code requires more physical qubits but has a higher code distance. Conversely, starting from a fixed symmetric code, one can identify different paired triorthogonal codes supporting transversal gates in different Clifford hierarchies, providing more flexibility for logical circuit optimization.

\begin{table}[]
\begin{tabular}{c|lll}
\hline
Triorthogonal Code & \multicolumn{3}{c}{Paired Symmetric Codes}            \\ \hline
$\mathcal{Q}_5^1\equiv[[31,1,3]]$         & \multicolumn{1}{c}{$\mathcal{Q}_3^1\equiv[[7,1,3]]$} & \multicolumn{1}{c}{$\mathcal{Q}_5^2\equiv[[31,1,7]]$} &  \multicolumn{1}{l}{}\\ 
$\mathcal{Q}_7^2\equiv[[127,1,7]]$         & \multicolumn{1}{c}{$\mathcal{Q}_5^2\equiv[[31,1,7]]$} & \multicolumn{1}{c}{$\mathcal{Q}_7^3\equiv[[127,1,15]]$} &
\multicolumn{1}{c}{} \\ 
$\mathcal{Q}_7^1\equiv[[127,1,3]]$         & \multicolumn{1}{c}{$\mathcal{Q}_3^1\equiv[[7,1,3]]$} & \multicolumn{1}{c}{$\mathcal{Q}_5^2\equiv[[31,1,7]]$} &  \multicolumn{1}{c}{$\mathcal{Q}_7^3\equiv[[127,1,15]]$}\\ \hline

\end{tabular}
\caption{Triorthogonal codes with paired symmetric codes.}
\label{Table.Examples}
\end{table}

\subsection{Symmetric Code Generation}

For triorthogonal codes that cannot be categorized as PQRM codes, we proposed another systematic method to generate a paired $\mathcal{Q}^{\rm Sym}$ that has one-way CNOT transversality. This generated symmetric code has the same physical qubit number and code distance as the triorthogonal code.

\begin{lemma}\label{Lem:symTtrans}
        \bc{Let $\mathcal{Q}^{\rm T}=\rm CSS(\mathcal{C}_1,\mathcal{C}_2)$ be a T-triorthogonal $[[n,k,d]]$ code, and let $\mathcal{C}$ be such that  $\mathcal{C}_1\subset \mathcal{C} \subset \mathcal{C}_2$. Then $\mathcal{Q}^{\rm T}$ and the symmetric $[[n,k,d']]$ CSS code $\mathcal{Q}^{\rm Sym}=\rm CSS(\mathcal{C},\mathcal{C})$ form a CNOT-transversal code pair, and $d'\ge d$.}
\end{lemma}

\begin{proof}

Since $\mathcal{C}_1\subset\mathcal{C}\subset\mathcal{C}_2$,
the generator matrix of $\mathcal{C}$ can be written as 
\begin{equation}
    \mathbf{G}_{\mathcal{C}}=\left[
    \begin{array}{c}
    \mathbf{G}\\
         \mathbf{B}
    \end{array}\right]=\left[
    \begin{array}{c}
    \mathbf{G}_1\\
    \mathbf{G}_0\\
         \mathbf{B}
    \end{array}\right],
    \label{Eq.GSym}
\end{equation}
where $\mathbf{G}$ is a triorthogonal matrix~\eqref{Eq:TriMatrix}, and $\mathbf{B}$ is a submatrix of $ \mathbf{G}_0^\perp$. Since $\mathcal{C}_2^{\perp}\subset\mathcal{C}^\perp\subset\mathcal{C}_1^\perp$, we can write the parity check matrix of $\mathcal{C}$, $\mathbf{G}_{\mathcal{C}}^\perp$, as follows
\begin{equation}
    \mathbf{G}_{\mathcal{C}}^\perp=\left[
    \begin{array}{c}
    \mathbf{G}_0\\
         \mathbf{D}
    \end{array}\right]\,.
    \label{}
\end{equation}
The restriction  $\mathcal{C}^\perp\subset\mathcal{C}$ of symmetric CSS codes implies that $\mathbf{G}_\mathcal{C}^\perp$ should be a self-orthogonal matrix. As $\mathbf{G}_1$ is not self-orthogonal, $\mathbf{D}$ has to be a submatrix of $\mathbf{B}$. 
According to the definition of the CSS code we have $\lvert
\mathcal{C}\rvert/\lvert\mathcal{C}^\perp\rvert=2^k$. Since the number of rows of $\mathbf{G}_1$ is $k$, we conclude that $\mathbf{D}=\mathbf{B}$.

Thus $\mathcal{C}_1 / \mathcal{C}_2^{\perp}\cong\mathbf{G}_1$ and $\mathcal{C} / \mathcal{C}^{\perp}\cong\mathbf{G}_\mathcal{C}/\mathbf{G}_\mathcal{C}^\perp=\mathbf{G}_1$.
So $\mathcal{C}_1 / \mathcal{C}_2^{\perp}\cong\mathcal{C} / \mathcal{C}^{\perp}$ and $\mathcal{C}_2^\perp\subset\mathcal{C}^\perp$, the codes fulfill the CNOT transversality conditions~\eqref{Eq:CNOTCondition}. We thus  have a transversal CNOT gate from $\mathcal{Q}^{\rm T}$ (control) to $\mathcal{Q}^{\rm Sym}$ (target). Since $\mathcal{C}\subset\mathcal{C}_2$ and $\mathcal{C}_1/\mathcal{C}_2^\perp\subset\mathcal{C}$, we also have CZ transversality.

The code distance of a  $\rm CSS(\mathcal{C}_i,\mathcal{C}_j)$ code is  $d=\min(d_i,d_j^{\perp})$, where $d_i$ and $d_j$ are the code distances of binary spaces $\mathcal{C}_i/\mathcal{C}_j^\perp$ and $\mathcal{C}_j/\mathcal{C}_i^\perp$. We already showed that $\mathcal{C}_1 / \mathcal{C}_2^{\perp}\cong\mathcal{C} / \mathcal{C}^{\perp}$ which indicates that  $d_1=d'$. For a symmetric code with distance $d'$, we have $d=\min(d',d_2)$, accordingly we have $d\leq d'$.
\end{proof}
\begin{theorem}
For any T-transversal CSS code $[[n, k, d ]]$ $\mathcal{Q}^{\rm T}$ with $n-k = 0\!\!\! \mod\! 2$, there exists a symmetric code $\mathcal{Q}^{\rm Sym}=\rm CSS(\mathcal{C},\mathcal{C})$ such that the pair is CNOT-transversal.
\end{theorem}
\begin{proof}
    Based on Lemma~\ref{Lem:symTtrans}, to construct $\mathcal{Q}^{\rm Sym}$ from $\mathcal{Q}^{\rm T}$, we need to find a self-orthogonal matrix $\mathbf{B}$ which is a submatrix of $\mathbf{G}_0^\perp$.
Since the dimension of $\mathbf{G}_\mathcal{C}^\perp$ is $\frac{n-k}{2}$,
the dimension of $\mathbf{B}$ is $\frac{n-k-2m}{2}\times n$, where $m$ is the number of rows in $\mathbf{G}_0$. As $\mathcal{C}_2^{\perp}\subset\mathcal{C}_1^\perp$ implies $m<n-m-k$, it follows that for any T-triorthogonal code with $n-k=0\mod2$, there exists a symmetric CSS code that satisfies the CNOT transversality conditions.
\end{proof}



As stated in Lemma~\ref{Lem:symTtrans}, to  construct the symmetric code $\mathcal{Q}^{\rm Sym}$ from $\mathcal{Q}^{\rm T}$, we just need to select a full-rank matrix $\mathbf{B}$ which is a submatrix of $\mathbf{G}_0^{\perp}$ with dimension $\frac{n-k-2m}{2}\times n$.

Here we use $[[3k+8,k,2]]$ triorthogonal codes~\cite{Bravyi_2012, Haah_2018} as an example to show the generated paired symmetric code. The $[[3k+8,k,2]]$ code family is not belongs to PQRM codes, and commonly used in magic state distillation protocol.

\begin{example}
    For $\mathbf{\rm T}$-transversal $[[14,2,2]]$ code which belongs to $[[3k+8,k,2]]$ codes family with triorthogonal matrix:
    \begin{align}
\mathbf{G}&=
\left[\begin{array}{c}
    \mathbf{G}_1  \\
    \hline
    \mathbf{G}_0
\end{array}\right] 
        =\left[ \begin{array}{cccccccccccccc}
0&0&0&0&1&1&1&1&1&1&1&0&0&0\\
0&0&0&0&1&1&1&1&0&0&0&1&1&1\\
\hline
0&1&0&1&0&1&0&1&1&0&1&1&0&1\\
0&0&1&1&0&0&1&1&0&1&1&0&1&1\\
1&1&1&1&1&1&1&1&0&0&0&0&0&0
         \end{array}\right]\nonumber.
\end{align}
We can find a symmetric code $\mathcal{Q}^{\rm sym}=\rm{CSS}(\mathcal{C}_3,\mathcal{C}_3)$ with $\mathcal{C}_3$ generated by matrix:
\begin{align}
\mathbf{G}_3&=
\left[\begin{array}{c}
    \mathbf{G}_0  \\
    \hline
    \mathbf{B}
\end{array}\right] 
        =\left[ \begin{array}{cccccccccccccc}
0&1&0&1&0&1&0&1&1&0&1&1&0&1\\
0&0&1&1&0&0&1&1&0&1&1&0&1&1\\
1&1&1&1&1&1&1&1&0&0&0&0&0&0\\
\hline
0&0&0&0&0&1&1&0&1&1&0&0&0&0\\
0&0&0&0&1&0&1&0&1&0&1&0&0&0\\
0&0&1&0&0&0&1&0&1&0&0&1&0&0
         \end{array}\right]\nonumber.
\end{align}
The $\mathcal{Q}^{\rm sym}$ is a $[[14,2,2]]$ symmetric code with Clifford gate transversality. 
\end{example}

By this example, we demonstrate that for any class of triorthogonal codes, compatible symmetric CSS partner codes can be systematically synthesized, thereby enabling transversal code switching between the two code blocks.







\section{State Teleportation Circuits}
\label{Sec.CircuitStructure}

In this section, we present circuits that enable logical state teleportation between two different code blocks, facilitating dynamic code switching during computations.

\subsection{Teleportation procedures}
The fault-tolerant teleportation circuit from $\mathcal{Q}^{\rm T}$ to $\mathcal{Q}^{{\rm Sym}}$ is shown in Fig.~\ref{fig:TxMapping} (a). It works as follows:
\bc{\begin{itemize}
    \item The input data block is encoded in $\mathcal{Q}^{\rm T}$ with an arbitrary state $\lvert\psi\rangle_L$.
    \item An ancilla block is prepared in the $\lvert0\rangle_L$ state with fault-tolerant manner, and encoded with $\mathcal{Q}^{{\rm Sym}}$.
    \item A round of $\mathbf{Z}$-basis stabilizer measurements are applied on $\mathcal{Q}^{\rm T}$, to eliminate high-weight $\mathbf{X}$ errors.
    \item A transversal CNOT gate is applied between corresponding $n$ physical qubits of $\mathcal{Q}^{\rm T}$ (control) and $\mathcal{Q}^{{\rm Sym}}$ (target).
    \item A logical $\mathbf{X}$-basis measurement is performed on the data block $\mathcal{Q}^{\rm T}$.
    \item Based on the measurement outcome, a logical Pauli $\mathbf{Z}$ operation is applied to the ancilla block. The resulting state in $\mathcal{Q}^{{\rm Sym}}$ is $\lvert\psi\rangle_L$, completing the logical state transfer.
\end{itemize}}

Because $\mathcal{Q}^{\rm T}$ can correct higher-weight $X$ errors than $\mathcal{Q}^{\mathrm{\rm sym}}$, such higher-weight errors must be diagnosed and corrected before code switching. Otherwise, they may induce an uncorrectable logical error on $\mathcal{Q}^{\mathrm{\rm sym}}$. This error-suppression step can be achieved by incorporating flag-qubit measurements into the FT state preparation circuit~\cite{chao2018faultflag,Chamberland_2018flag}.

The reverse teleportation, from $\mathcal{Q}^{{\rm Sym}}$ to $\mathcal{Q}^{\rm T}$, cannot be achieved using the same circuit, as the logical CNOT in that direction is not transversal. Instead, we use the transversal logical CZ gate between these two codes. The corresponding teleportation circuit is shown in Fig.~\ref{fig:TxMapping}(b), and the procedure is as follows:
\begin{itemize}
    \item The input data block is encoded in $\mathcal{Q}^{\rm Sym}$ with an arbitrary state $\lvert\psi\rangle_L$. 
    \item An ancilla block is fault-tolerantly prepared in the logical $\lvert+\rangle_L$ state, encoded with $\mathcal{Q}^{{\rm T}}$. Notice that additional flag qubits are introduced compared to the normal FT state preparation circuit to erase high-weight $\mathbf{X}$ errors.
    \item Transversal logical Hadamard gates are applied to the data block.
    \item Transversal CZ gates are applied between $\mathcal{Q}^{{\rm Sym}}$ and $\mathcal{Q}^{\rm T}$.
    \item A logical $\mathbf{X}$-base measurement is performed on the data block $\mathcal{Q}^{\rm Sym}$. Based on the measurement outcome, a logical Pauli $\mathbf{X}$ operation is applied to $\mathcal{Q}^{\rm T}$.
    \item  The final state in $\mathcal{Q}^{\rm T}$ is $\lvert\psi\rangle_L$, completing the logical teleportation.
\end{itemize}

For the case $\mathcal{Q}^{{\rm Sym}}$ having a higher code distance than $\mathcal{Q}^{{\rm T}}$, we need to conduct an extra round of $\mathbf{X}$-basis stabilizer measurement to $\mathcal{Q}^{\rm Sym}$ before performing transversal CZ gates, to eliminate high-weight $\mathbf{Z}$ errors. However, when we choose the $\mathcal{Q}^{{\rm Sym}}$ with the same code distance as $\mathcal{Q}^{\rm T}$, this step is no longer necessary.

With these two circuits, we obtain a two-way teleportation between $\mathcal{Q}^{\rm T}$ and $\mathcal{Q}^{\rm Sym}$. During computation, logical states can be stored in $\mathcal{Q}^{\rm Sym}$ while Clifford operations are executed transversally. When non-Clifford gates are required, the logical state is teleported back to $\mathcal{Q}^{\rm T}$, where such operations can also be applied transversally. 

Although each teleportation circuit requires one ancilla code block, this overhead can be reduced by integrating the teleportation procedure into the Steane syndrome extraction process. We provide more details on this in the following section.
\begin{figure}
\centering
\includegraphics[width=0.8\linewidth]{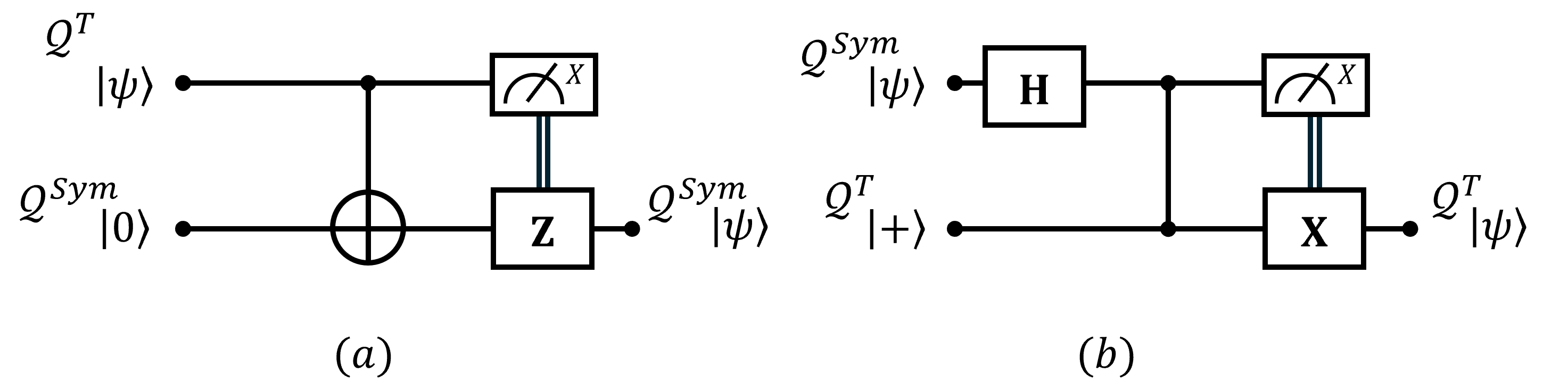}
\caption{The state teleportation circuits between $\mathcal{Q}^{\rm T}$ and $\mathcal{Q}^{\rm Sym}$. (a) The state teleportation from $\mathcal{Q}^{\rm T}$ to $\mathcal{Q}^{\rm Sym}$. (b) Since the CNOT gate from $\mathcal{Q}^{\rm Sym}$ to $\mathcal{Q}^{\rm T}$ is not transversal, we can use CZ gate instead. The teleportation circuit is similar, but with additional logical Hadamard gate at beginning, and this logical Hadamard gate can be applied transversally on $\mathcal{Q}^{\rm Sym}$.}
\label{fig:TxMapping}
\end{figure}

\subsection{Steane Error Correction Circuit}
\label{sec.Steaneframe}

In this section, we show how to merge transversal code switching protocol into Steane error correction procedure. 

The logical state teleportation circuits from $\mathcal{Q}^{\rm T}$ to $\mathcal{Q}^{\rm Sym}$ shown in Fig.~\ref{fig:TxMapping}(a) can be embedded into the syndrome extraction process in a manner that preserves fault-tolerance. The teleportation circuit in Fig.~\ref{fig:TxMapping}(a), merged with the Steane error correction procedure is shown in Fig.~\ref{fig:StateTele}. 
The circuit on the right represents the process in which the logical state, initially encoded in $\mathcal{Q}^{\rm T}$ is teleported into $\mathcal{Q}^{\rm Sym}$ while simultaneously performing one round of error correction. Similarly, the merged version of the circuit in Fig.~\ref{fig:TxMapping}(b) is shown in Fig.~\ref{fig:SymtoTQEC}, where the logical state is teleported from $\mathcal{Q}^{\rm Sym}$ to $\mathcal{Q}^{\rm T}$, followed by a round of error correction. 

The $\mathbf{Z}$-base measurements on the first ancilla block yield error syndromes that can identify and correct high-weight physical $\mathbf{X}$ errors in the original code. 
However, since the transversal CNOT/CZ gates are applied to the second ancilla block, any $\mathbf{X}$ errors on the data block will also propagate to the second ancilla block. 
Therefore, recovery Pauli $\mathbf{Z}$ operations must be applied to the second ancilla block, along with the corresponding Pauli basis.
For the situation where we switch from $\mathcal{Q}^{{\rm T}}$ to $\mathcal{Q}^{{\rm Sym}}$, the extra round of $\mathbf{Z}$-basis stabilizer measurement is equivalent to this $\mathbf{Z}$-basis measurement since it could also correct the high-weight $\mathbf{X}$ errors that can cause a logical error in $\mathcal{Q}^{{\rm Sym}}$.
 

The $\mathbf{X}$-base 
measurements on the data block yield a classical binary vector corresponding to a (possibly corrupted) codeword in $\mathcal{C}_2$.
Using the extracted syndrome, $\mathbf{Z}$ errors on the data can be identified and corrected. 
Since $\mathbf{Z}$ errors commute with the CZ gates, they do not propagate to the second ancilla block, which stores in the data block. 
After correction, the resulting binary vector is projected back into $\mathcal{C}_2$, and the corresponding logical Pauli operator is used to determine whether the logical state is $\lvert+\rangle_L$ or $\lvert-\rangle_L$. If the result is $\lvert-\rangle_L$ state, a logical Pauli operator is applied as feedback. 

Using Steane's error correction framework will reduce the cost of CNOT gates and also remove the need of extra $\mathbf{Z}$ or $\mathbf{X}$-basis stabilizer measurements, but it requires an additional ancilla code block.
Alternatively, one can use flag-based fault-tolerant method~\cite{Chamberland_2018flag,chao2020flag} or Shor's error correction method~\cite{shor1996fault}, which they can reduce the qubit number requirement but increase the CNOT gates number.   

\begin{figure}
    \centering
    \includegraphics[width=0.9\linewidth]{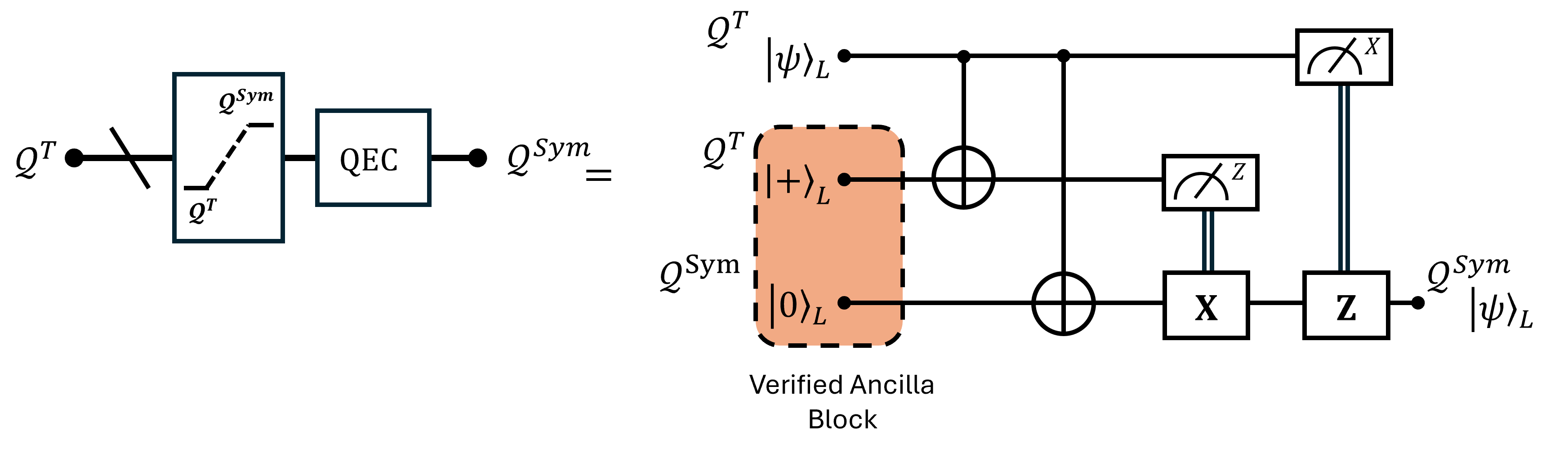}
    \caption{The state teleportation circuit from $\mathcal{Q}^{\rm T}$ to $\mathcal{Q}^{\rm Sym}$, embedded with Steane syndrome extraction procedure. After the circuit is executed, the logical state is re-encoded in $\mathcal{Q}^{\rm Sym}$, and one round of error correction has been completed. }
    \label{fig:StateTele}
\end{figure}
\begin{figure}
    \centering
    \includegraphics[width=0.9\linewidth]{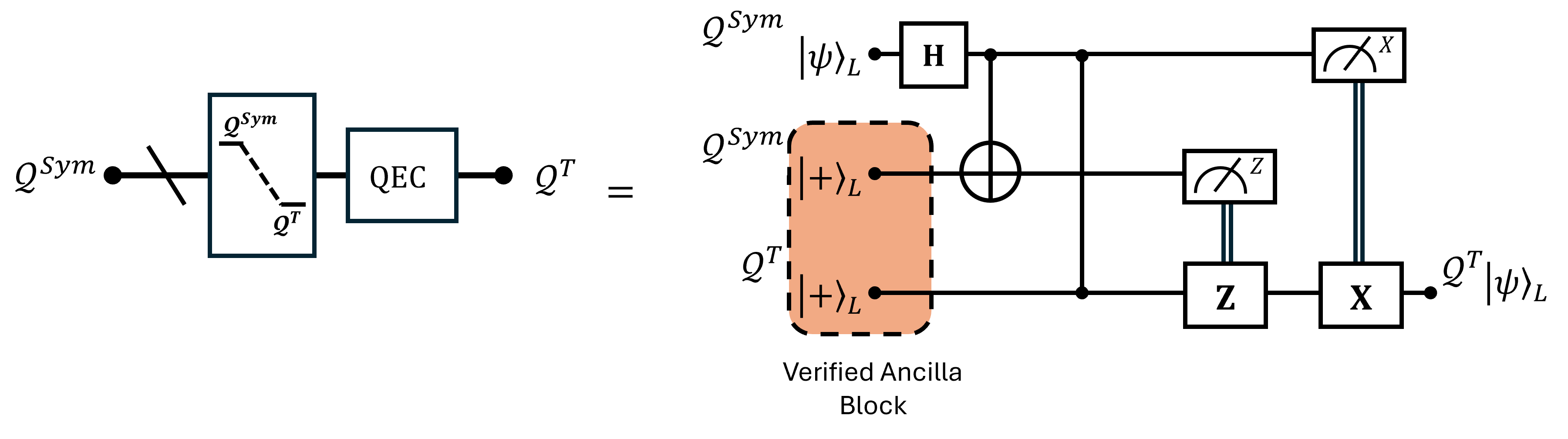}
    \caption{The state teleportation circuit from $\mathcal{Q}^{\rm Sym}$ to $\mathcal{Q}^{\rm T}$, embedded with Steane syndrome extraction procedure. 
    After the circuit is executed, the logical state is re-encoded in $\mathcal{Q}^{\rm T}$, and one round of error correction has been simultaneously completed.}
    \label{fig:SymtoTQEC}
\end{figure}

\section{Resource Cost}
\label{Sec.ResourceEST}

In this section, we evaluate the resource overhead of the transversal code-switching protocol for triorthogonal codes. Taking the $[[15,1,3]]$ and $[[31,1,3]]$ triorthogonal codes as representative examples, we quantify the CNOT gate count required by the Steane error correction method for transversal code switching. Furthermore, we apply these switching protocols to implement a fault-tolerant logical quantum Fourier transform (QFT) circuit. By comparing the overall resource cost for both code families, we demonstrate the CNOT gate-count advantages of using triorthogonal codes that support transversal gates at higher levels of the Clifford hierarchy.


\subsection{Cost of Transversal Code Switching}

In this section, we evaluate the resource cost of performing code switching together with one full round of error correction and compare the total overhead of using $[[15,1,3]]$ and $[[31,1,3]]$ codes. Consider that we use the Steane error correction method discussed in~\ref{sec.Steaneframe}. 

The resource cost mainly from two section, the fault-tolerant ancilla state preparation and transversal gates. 
For the fault-tolerant ancilla state preparation, one could use non-fault-tolerant encoding with the flag-qubit verification to achieve full fault-tolerance~\cite{chao2020flag}. Based on the method and algorithms proposed in~\cite{mqt,zen_2025}, we found the optimal $[[31,1,3]]$ FT $\lvert +\rangle_L$ state preparation circuit, that use 69 CNOT gates and 1 ancilla qubit to prepare logical $\lvert +\rangle_L$ state and 12 CNOT gates with 3 flag qubits to erase high-weight $\mathbf{X}$ errors, the circuit shown as in Fig.~\ref{fig:31qCircuit}. The $[[15,1,3]]$ code FT state preparation circuits are also shown in~\cite{zen_2025}, which need 25 CNOT gates with 1 ancilla qubit for $\lvert0\rangle_L$ state and 31 CNOT gates with 1 ancilla qubit for $\lvert+\rangle_L$ state.
As for the Steane code, the cost of state preparation is 11 CNOT gates with 1 ancilla qubit for both $\lvert0\rangle_L$ and $\lvert+\rangle_L$ states.

\begin{figure}
    \centering
    \includegraphics[width=\linewidth]{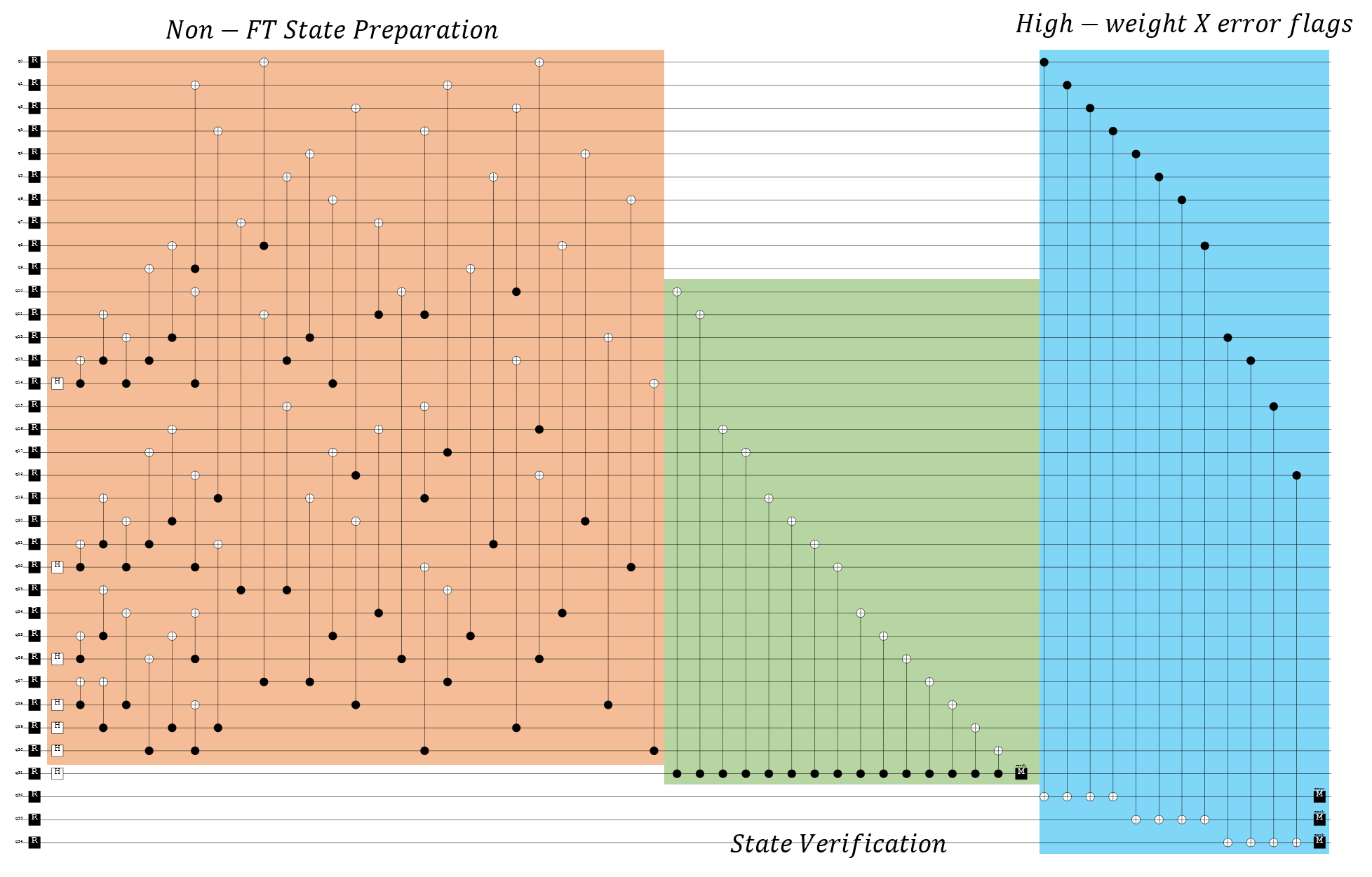}
    \caption{The $[[31,1,3]]$ code logical $\lvert+\rangle_L$ state preparation circuit, with three flag qubits for high-weight $\mathbf{X}$ error detection. }
    \label{fig:31qCircuit}
\end{figure}

The cost of transversal two-qubit gates in Steane error correction consists of one-way CNOT to Steane code and transversal CNOT to the ancilla code block encoded with its own code; for $[[15,1,3]]$ code, the cost is $15+7=22$ CNOT gates, and for $[[31,1,3]]$ code, it is $31+7=38$ CNOT gates. The overall resource cost to switch from $[[15,1,3]]$ code to Steane code and to complete one round of QEC is $39$ qubits with $64$ CNOT gates; reverse direction switching and QEC need $31$ qubits and $30$ CNOT gates. If we use $[[31,1,3]]$ code instead, the overall resource cost is $74$ qubits with $130$ CNOT gates from $[[31,1,3]]$ to Steane code and $50$ qubits with $106$ CNOT gates.  

\subsection{Quantum Fourier Transform Circuit}

Here we pick the Quantum Fourier Transform (QFT) circuit as an example to show that using $[[31,1,3]]$ code with $\sqrt{\mathbf{\rm T}}$ gate transversality saves overall resource as compared to a $[[15,1,3]]$ code with $\mathbf{\rm T}$ gate transversality.
The QFT circuit is a fundamental component in many quantum algorithms, such as Shor's factorization and Harrow–Hassidim–Lloyd (HHL) algorithm. Here we use transversal code switching protocol to achieve Clifford $ +\, \mathbf{\rm T}$ or Clifford $ +\, \mathbf{\sqrt{\rm T}}$ gate sets to decompose the n-qubit QFT circuit. We assume that $\mathrm{T}$ and 
$\mathbf{\sqrt{\rm T}}$ -gates are executed in $[[15,1,3]]$ and $[[31,1,3]]$ -code, respectively, where the are transversal, whereas Hadamard-gates and multi-controlled CNOTs are executed in the $[[7,1,3]]$ Steane code, where they are transversal.  

The standard n-qubit QFT circuit uses controlled-rotation (CR) gates with different angles as the primary non-Clifford gates in the circuit:
\begin{equation}
    \mathbf{CR}(n,j)=\frac{\mathbf{I}+\mathbf{Z}}{2}\otimes \mathbf{I}+\frac{\mathbf{I}-\mathbf{Z}}{2}\otimes \mathbf{R}_z\left(\frac{2\pi i}{2^{n-j}}\right),
\end{equation}
where $0<j<n$ is an integer. The n-qubit QFT circuit uses $j$ $\mathbf{CR}(n,j)$ gates and a total of $\frac{n(n-1)}{2}$ CR gates. Using the decomposition proposed in~\cite{park2024tcountoptimizationapproximatequantum}, we can decompose the $n$-qubit QFT circuit with CR gates into CNOT and $\mathbf{R}_Z$ gates, as shown in Fig.~\ref{fig:QFTcircuit}. 
Then, we can obtain the $\mathbf{\rm T}$-gate count of the $n$-qubit QFT circuit by decomposing $\mathbf{R}_Z$ gates into  Clifford $+\,\mathbf{\rm T}$ gates with the given approximation precision $\varepsilon$. For the decomposition, we use the algorithm proposed in~\cite{ross2016optimalancillafreecliffordtapproximation}. As for Clifford $+\,\sqrt{\mathbf{\rm T}}$ gate decomposition, we directly use the conclusion of~\cite{Kliuchnikov_2023}, which indicates that for diagonal gates, on average, one needs $(3.03\log_2(1/\varepsilon)+2.48)/4$ $\sqrt{\mathbf{\rm T}}$ gates in the decomposition. For both cases, we set accuracy $\varepsilon=10^{-10}$. The total number of CNOT gates has two components, one coming from the decomposition, and one coming from the algorithm. We use the total number of CNOT gates as a proxy of the error performance of the circuit, assuming that two-qubit gates dominate performance. 

\begin{figure}
    \centering
    \includegraphics[width=\linewidth]{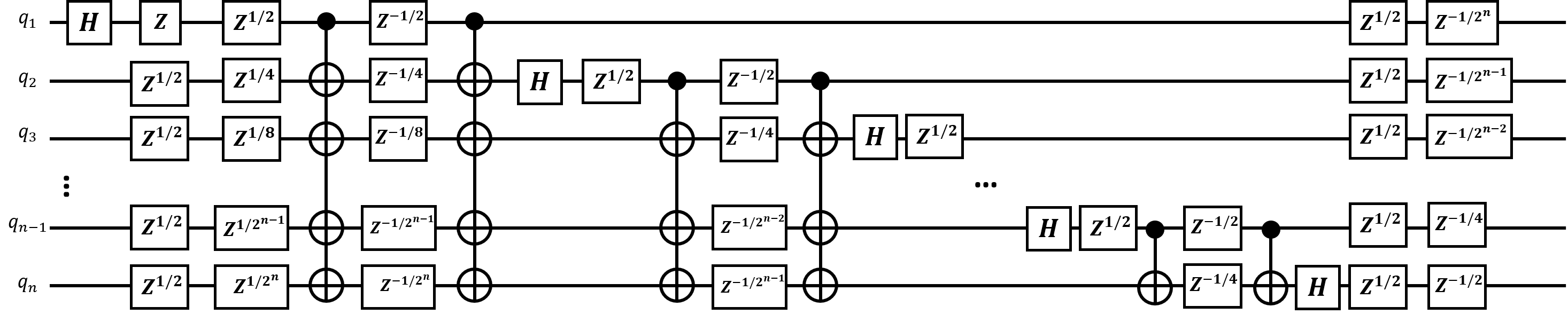}
    \caption{The decomposed n-qubit QFT circuit, using the method proposed in~\cite{park2024tcountoptimizationapproximatequantum}.}
    \label{fig:QFTcircuit}
\end{figure}

As illustrated in Fig.~\ref{fig:QFT}, we compare the CNOT gate overhead required to implement $n$-qubit QFT circuits across varying qubit numbers $n$, utilizing both the $[[15,1,3]]$ and $[[31,1,3]]$ codes in combination with the Steane code.In this analysis, we assume that physical single-qubit gates (such as $\mathbf{\rm T}$, $\sqrt{\mathbf{\rm T}}$, and Clifford gates) can be implemented with minimal overhead; thus, we focus exclusively on the CNOT gate cost. The total CNOT gate count is driven by two primary components, transversal code switching and logical CNOT gate implementation. For transversal code switching, a full round of bi-directional code switching is performed after every sequence that transitions from Clifford to non-Clifford gates and back.
In the logical $n$-qubit QFT circuit, it inherently requires $\frac{n(n-1)}{2}$ logical CNOT gates. To minimize overhead, we first switch the target code blocks to the Steane code before applying a transversal CNOT gate. Because the Steane code has 7 physical qubits, all the logical CNOT gates adds $\frac{7n(n-1)}{2}$ physical CNOT gates to the total count.

As shown in the figure, for large values of $n$, utilizing the $[[31,1,3]]$ code reduces the CNOT gate cost by approximately half compared to the $[[15,1,3]]$ code. For QFT circuits requiring a smaller number of qubits, the $[[31,1,3]]$ code achieves an even more significant reduction, lowering the CNOT gate count by a factor of 3 to 10.It is important to note that these costs are calculated based on a specific gate decomposition method for the QFT circuit~\cite{park2024tcountoptimizationapproximatequantum}. Furthermore, because synthesis algorithms for the $\sqrt{\mathbf{\rm T}}$ gate are currently less explored than the extensively optimized routines for the standard $\mathbf{\rm T}$ gate, there remains substantial potential to further optimize the $\sqrt{\mathbf{\rm T}}$-gate count and, consequently, reduce the overall CNOT overhead even further.

\begin{figure}
    \centering
    \includegraphics[width=0.7\linewidth]{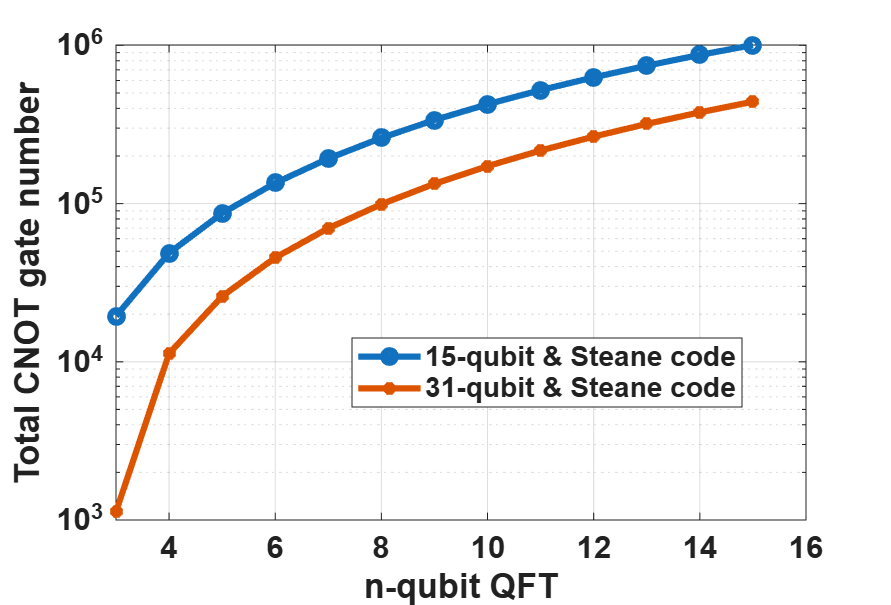}
    \caption{The CNOT gate cost for applying $n$-qubit QFT circuit.}
    \label{fig:QFT}
\end{figure}

\section{Conclusion}

In this work, we introduce transversal code switching protocols for triorthogonal codes. Depending on the underlying code structure, our approach  divides into two strategies. The first strategy targets the subclass of triorthogonal codes belonging to the PQRM family. Focusing on the single-logical qubit regime ($k=1$), we prove the conditions required for one-way transversality across specific PQRM parameter regimes. Building on this theoretical foundation, we construct a \emph{code map} that systematically identifies compatible, symmetric PQRM code partners for any given triorthogonal PQRM code.
For general triorthogonal codes outside the PQRM family, we developed a systematic algorithm to construct their symmetric code pairs and execute one-way transversal CNOT operations. By pairing a triorthogonal code with its symmetric code, logical information can be transferred using a fully transversal state teleportation circuit. This transversal code switching mechanism successfully establishes a universal, fault-tolerant gate set. 
Furthermore, we demonstrated that integrating Steane error correction into the transversal switching protocol significantly reduces the overall gate overhead. Finally, to validate the advantages of larger PQRM codes, we analyzed the performance of our transversal switching method in compiling $n$-qubit QFT circuits. A comparative analysis of the CNOT overhead between the $[[15,1,3]]$ code and the $[[31,1,3]]$ code, both paired with the Steane code, revealed that using transversal operations at higher levels of the Clifford hierarchy yields substantial reductions in circuit complexity. These results show the practical viability and algorithmic advantage of deploying high-rank PQRM codes in large-scale quantum computation.

\bibliographystyle{quantum}
\bibliography{apssamp}

\onecolumn

\end{document}